\documentclass[a4paper,english,11pt]{article}

\usepackage[compact]{titlesec}
\usepackage{booktabs}   
\usepackage{subcaption} 
\usepackage{xspace}                        
\usepackage{multicol}

\usepackage{tabularx, environ}
\usepackage{mathbbol}
\usepackage[ruled, vlined,  commentsnumbered,linesnumbered]{algorithm2e}
\usepackage{adjustbox}
\usepackage{mathtools}
\usepackage{multirow}
\usepackage{enumitem}
\usepackage{fancyvrb}
\usepackage{tensor}
\usepackage{listings}
\usepackage[numbers,sort,square,comma]{natbib}
\usepackage{tikz}
\usetikzlibrary{arrows}
\usetikzlibrary{calc}
\usetikzlibrary{positioning}
\usepackage{semantic}
\usepackage{color} 

\newenvironment{spmatrix}[1]
 {\def\mysubscript{#1}\mathop\bgroup\begin{pmatrix}}
 {\end{pmatrix}\egroup_{\textstyle\mathstrut\mysubscript}}
\newcommand{\np}{${\bf{NP}}$\xspace}

\newcommand{\greynode}{{\color{gray} \bullet}}

\newcommand{\npda}{${\bf{2NPDA}}$\xspace}

\newcommand{\myie}{\emph{i.e.}\xspace}

\newcommand{\myeg}{\emph{e.g.}\xspace}

\newcommand{\mywrt}{\emph{w.r.t.}\xspace}

\newcommand{\codeIn}[1]{{\small\tt{#1}}}
\newcommand{\codeInM}[1]{{{#1}}}

\makeatletter

\newcolumntype{\expand}{}
\long\@namedef{NC@rewrite@\string\expand}{\expandafter\NC@find}

\NewEnviron{problem}[2][]{%
  \def\problem@arg{#1}%
  \def\problem@framed{framed}%
  \def\problem@lined{lined}%
  \def\problem@doublelined{doublelined}%
  \ifx\problem@arg\@empty%
    \def\problem@hline{}%
  \else%
    \ifx\problem@arg\problem@doublelined%
      \def\problem@hline{\hline\hline}%
    \else%
      \def\problem@hline{\hline}%
    \fi%
  \fi%
  \ifx\problem@arg\problem@framed%
    \def\problem@tablelayout{|>{\bfseries}lX|c}%
    \def\problem@title{\multicolumn{2}{|l|}{%
        \raisebox{-\fboxsep}{\textsc{\Large #2}}%
      }}%
  \else
    \def\problem@tablelayout{>{\bfseries}lXc}%
    \def\problem@title{\multicolumn{2}{l}{%
        \raisebox{-\fboxsep}{\textsc{\Large #2}}%
      }}%
  \fi%
  \bigskip\par\noindent%
  \begin{tabularx}{\textwidth}{\expand\problem@tablelayout}%
    \problem@hline%
    \problem@title\\[2\fboxsep]%
    \BODY\\\problem@hline%
  \end{tabularx}%
  \medskip\par%
}
\makeatother

\usepackage{amsthm}
\usepackage{ amssymb }
\newtheorem{conjecture}{Conjecture}
\newtheorem{theorem}{Theorem}
\newtheorem{corollary}{Corollary}
\newtheorem{definition}{Definition}
\newtheorem{example}{Example}
\newtheorem{lemma}{Lemma}
\usepackage{geometry}
\usepackage{hyperref}

\begin{document}

\title{Conditional Lower Bound for Inclusion-Based Points-to Analysis}         



\author{Qirun Zhang\thanks{Georgia
    Institute of Technology, School of Computer Science, {\tt qrzhang@gatech.edu}. }} 

\maketitle
\begin{abstract}


Inclusion-based (\myie, Andersen-style) points-to analysis  is a fundamental static
analysis problem. The seminal work of Andersen gave a worst-case cubic $O(n^3)$
time points-to analysis algorithm for C, where $n$ is proportional to the number of
variables in the input program. 
From an algorithmic perspective, an algorithm is \emph{truly subcubic} if
it runs in $O(n^{3-\delta})$ time
 for some $\delta > 0$.
Despite decades of extensive effort on improving
points-to analysis, 
the cubic bound remains unbeaten.
The best combinatorial analysis
algorithms have a ``slightly subcubic'' $O(n^3 /
\text{log } n)$ complexity which improves Andersen's original algorithm by only
a log factor.
It is an interesting open problem whether inclusion-based points-to analysis can
be solved in truly subcubic time.

In this paper, we prove that a truly subcubic $O(n^{3-\delta})$ time
combinatorial algorithm for inclusion-based points-to
analysis is unlikely:  a truly subcubic combinatorial points-to analysis algorithm  implies a truly subcubic combinatorial algorithm for Boolean Matrix
Multiplication (BMM). BMM is a well-studied problem, 
and no truly subcubic combinatorial BMM algorithm 
has been known.
The fastest
combinatorial BMM algorithms  run in time
$O(n^3/ \text{log}^4 n)$.


Our conditional lower bound result  includes a
simplified proof of the BMM-hardness of Dyck-reachability. The reduction is
interesting in its own right. First, it is slightly stronger than the existing
BMM-hardness results of Dyck-reachability because our reduction only requires
one type of parenthesis in Dyck-reachability ($D_1$-reachability). Second, we
formally attribute the ``cubic bottleneck'' of points-to analysis to
the need to solve $D_1$-reachability, which captures the  semantics of properly-balanced pointer
references/dereferences. This new perspective enables a more general
reduction that applies to programs with arbitrary types of pointer statements.
Last, our reduction based on $D_1$-reachability shows that demand-driven points-to
analysis is as hard as the exhaustive counterpart. The hardness result 
generalizes to a wide variety of demand-driven interprocedural program analysis
problems.

\end{abstract}







\section{Introduction}\label{sec:intro}
Points-to analysis is a fundamental static analysis problem. Points-to
information is a prerequisite for many practical program
analyses. Points-to analysis is computationally hard. 
It is well-known that computing the precise points-to information is undecidable~\cite{Ramalingam94the}. Even for the simplest (\myie,
context- and flow-insensitive) variant, the precise analysis problem is known
to be \np-hard~\cite{Horwitz97precise}. 
Any practical points-to analysis must approximate the exact solution. In the
literature, two predominant frameworks for computing  sound points-to information  are
equality-based (\myie, Steensgaard-style)~\cite{Steensgaard96points} and 
inclusion-based (\myie, Andersen-style) points-to analyses~\cite{andersen1994program}.



Inclusion-based points-to analysis~\cite{andersen1994program} is more precise than equality-based
analysis.  A study by \citet{BlackshearCSS11the} shows that ``The
precision gap between Andersen's and precise flow-insensitive analysis is non-existent in practice.''
An inclusion-based points-to analysis collects a set of inclusion constraints
from the input program and constructs a constraint graph.
It obtains the points-to information by computing a fixed point solution over the
corresponding constraint graph~\cite{FahndrichFSA98partial,HardekopfL07the,SridharanF09the}.
All existing points-to analysis algorithms are \emph{combinatorial} in the sense
that they are discrete and graph-theoretic.
However, computing the closure of the constraint graph is quite expensive.
Despite decades of
research, the fastest algorithm for inclusion-based
points-to analysis  exhibits an ``slightly subcubic'' $O(n^3 / \text{log }n)$ time
complexity~\cite{SridharanF09the,Chaudhuri08subcubic}. 
In practice, many studies have observed a quadratic scaling behavior for inclusion-based
pointer analysis~\cite{SridharanF09the}. 
It is open whether inclusion-based points-to analysis problem admits \emph{truly
subcubic} algorithms, \myie, algorithms with running time $O(n^{3-\delta})$ for some constant
$\delta >0$. 

In this paper, we prove a conditional lower
bound for the inclusion-based points-to analysis problem, which shows that a truly subcubic
combinatorial algorithm is unlikely to exist. 
Our hardness result is based on
the popular Boolean Matrix Multiplication (BMM) conjecture:
\begin{conjecture}[Boolean Matrix Multiplication~\cite{AbboudW14popular,WilliamsW18subcubic}]\label{conj:bmm}
For all
$\delta >0$, there exists no combinatorial algorithm that computes the product of two
$n\times n$ Boolean matrices in  time $O(n^{3-\delta})$.
\end{conjecture}The conjecture states 
that in the
RAM model with $O(\log n)$ bit words, any combinatorial  BMM algorithm  requires
$n^{3-o(1)}$ time~\cite{ww18onsome}. 
The BMM conjecture has been utilized to prove  fine-grained lower bounds
of many problems in theoretical computer science~\cite{HenzingerKNS15unifying,AbboudW14popular,WilliamsW18subcubic} and program
analysis/verification~\cite{ChatterjeeCP18opt,ChatterjeeDHL16model,Lee02fast}. 
Note that BMM can be solved in truly subcubic time using the heavy machinery of fast matrix multiplication (FMM)
 originated by Strassen~\cite{strassen1969gaussian}.
The current fastest algorithms run in $O(n^{2.373})$
time~\cite{Gall14apowers,Williams12mult}.
However, those algorithms based on
FMM are \emph{algebraic} which rely on the ring structure of matrices over the
field.  Algebraic approaches have enjoyed little success in
practice~\cite{AbboudBW18if,Yu15an}.
In particular, we are unaware of any practical program analyses that are based on FMM.
On the other hand,
inclusion-based points-to analysis has been
formulated as  a combinatorial problem of resolving inclusion constraints (via a
sequence of set-union and table-lookup operations). All existing pointer
analysis algorithms are combinatorial  in nature. 
Moreover,  combinatorial
algorithms are practical and
avoid FMM or other
``Strassen-style''
methods.\footnote{In many combinatorial structures such as Boolean semiring, there is no
inverse under addition. In practice, algebraic methods have large
hidden constants and generally considered  impractical. See more detailed
discussions in the work of \citet{BallardDHS12graph}, \citet{BansalW09regular} and \citet{HenzingerKNS15unifying}.}
Thus, lower bounds for combinatorial algorithms are of particular interest from
both theoretical and practical perspectives. 

Our proof of the BMM-hardness of inclusion-based points-to analysis involves two
steps. In the first step, we reduce BMM to Dyck-Reachability with
one type of parentheses ($D_1$-reachability).
Dyck-reachability is a graph reachability problem where the edges in the input
graph are labeled with $k$ types of open and close 
parentheses~\cite{Reps97program,ChatterjeeCP18opt,qirun13fast}. The goal is to compute all reachable nodes that
can be joined by paths with properly-matched parentheses.
Dyck-reachability has been utilized to express many static analysis problems, such
as interprocedural data flow analysis~\cite{RepsHS95precise}, program
slicing~\cite{RepsHSR94speeding}, shape analysis~\cite{Reps95shape}, and type-based flow analysis~\cite{PratikakisFH06exist,RehofF01type}. It has also been widely used in practical analysis tools such as Soot~\cite{Bodden12ainter}.
In the second step of our proof, we reduce $D_1$-reachability to inclusion-based
points-to analysis. In particular, we introduce a novel Pointer Expression Graph (PEG)
representation for C-style programs and formulate the points-to analysis as a
$\mathit{Pt}$-reachability problem on PEGs.
The key insight in our reduction is to leverage the pointer reference and dereference in
$\mathit{Pt}$-reachability formulation to express the properly-matched parentheses in $D_1$-reachability.
Our two-step reduction yields two conditional lower bounds on both
$D_1$-reachability and inclusion-based points-to analysis:

\begin{theorem}[BMM-hardness of Dyck-reachability]\label{thm:bmm1}
For any fixed $\delta >0$, if there is a combinatorial algorithm that solves $D_1$-reachability  in $O(n^{3-\delta})$ time, then there is a
combinatorial algorithm that solves Boolean Matrix Multiplication in truly
subcubic time.
\end{theorem}

\begin{theorem}[BMM-hardness of inclusion-based points-to analysis]\label{thm:bmm2}
For any fixed $\delta >0$, if there is a combinatorial algorithm that solves
inclusion-based points-to analysis in $O(n^{3-\delta})$ time, then there is a
combinatorial algorithm that solves Boolean Matrix Multiplication in truly
subcubic time.
\end{theorem}

Our key insight on $D_1$-reachability-based reduction yields several interesting
implications. 
We formally summarize  the results as two corollaries.
In particular, Corollary~\ref{cor:uni} considers any \emph{non-trivial} C-style programs
with pointers and Corollary~\ref{cor:dd} applies to programs with pointer
dereferences.\footnote{Non-trivial C-style programs always contain  address-of
  statements 
  ``\codeIn{a = \&b}'' and at least one of the three types of statements ``\codeIn{a =
  b}'', ``\codeIn{*a = b}'', and ``\codeIn{a = *b}''. The address-of statements ``initialize'' the points-to
  sets of all variables.  Without such statements, all points-to sets in the
  program  are empty sets. For programs with only address-of statements, all
  points-to sets can be obtained via a simple linear-time scan.  Both types of programs are trivial for  pointer analysis.}

\begin{corollary}[Universality]\label{cor:uni}
Inclusion-based points-to analysis is BMM-hard for non-trivial programs under
 arbitrary combinations of  statement types
``\codeIn{a =
  b}'', ``\codeIn{a = \&b}'', ``\codeIn{*a = b}'', and ``\codeIn{a = *b}''.
\end{corollary}

\begin{corollary}[BMM-hardness of demand-driven analysis]\label{cor:dd}
In the presence of pointer dereferences, the demand-driven inclusion-based
points-to analysis is
BMM-hard.
\end{corollary}

The two corollaries  offer a comprehensive  view on the cubic
bottleneck of inclusion-based points-to analysis.
Prior to our work, 
it is
  folklore that the complexity of inclusion-based points-to analysis is
  related to computing the ``dynamic transitive closure'' (DTR).
Unfortunately, the ``dynamic'' aspect of DTR, which informally refers to the process of adding inclusion-constraint edge to
the graph during constraint resolution, has not been rigorously
defined. 
The work by
\citet{SridharanF09the} establishes a relation between points-to analysis
and  transitive closure. However, the reduction is  restrictive  because it
applies to  programs with no pointer
dereferences (\myie, programs consist of only two types of statements ``\codeIn{a =
  b}'' and ``\codeIn{a = \&b}''). 
Our universality corollary (Corollary~\ref{cor:uni}), which generalizes to all
types of
program statements, is more
principled. Specifically, 
the $D_1$-reachability used in our reduction formally captures the properly-matched pointer references/dereferences.
Moreover, this insight also enables Corollary~\ref{cor:dd} which shows that demand-driven pointer analysis is as
hard as the exhaustive counterpart.
Note 
the traditional transitive-closure-based reduction  can not yield
Corollary~\ref{cor:dd}   as the single-source-single-target graph
reachability can be trivially solved in $O(n)$ time via a depth-first search.

To sum up, the significance of our results is threefold. 

\begin{itemize}
\item Theorem~\ref{thm:bmm1} gives a simplified proof to establish the BMM-hardness of
  Dyck-reachability. It is slightly stronger than the existing hardness results
  because previous proofs ~\cite{ChatterjeeCP18opt,HeintzeM97on}
 require Dyck languages of $k \geq 2$ types of parentheses.
\item Theorem~\ref{thm:bmm2} formally establishes a cubic lower bound for
  inclusion-based points-to analysis conditioned on the popular BMM conjecture.
\item The proof of Theorem~\ref{thm:bmm2} is based on the reduction from
  $D_1$-reachability. The key insight of depicting the well-balanced pointer
  references/dereferences with $D_1$-reachability yields
  two interesting corollaries:
\begin{itemize}
\item Corollary~\ref{cor:uni} permits all types of 
  constraints in inclusion-based  points-to analysis and makes no assumption of
  the input programs.  Previous transitive-closure-based reduction~\cite{SridharanF09the} is 
  restricted to programs with no pointer dereferences, which cannot be
  generalized to programs with pointer dereferences. 
\item Corollary~\ref{cor:dd} generalizes to interprocedural program analysis. It demonstrates that the bottleneck of demand-driven interprocedural
analysis is due to matching the well-balanced properties such as procedure
calls/returns and pointer references/dereferences, as opposed to computing the
transitive closure.
\end{itemize}
\end{itemize}

The rest of the paper is organized as follows. Section~\ref{sec1:pred}
introduces problem definitions. Section~\ref{sec1:over} presents an overview of our
reductions.
Sections~\ref{sec1:pt1} and \ref{sec1:pt2} prove the reduction correctness. Section~\ref{sec:imp} discusses the implications. Section~\ref{sec1:rw} surveys related work and
Section~\ref{sec1:conc} concludes.

\section{Preliminaries} \label{sec1:pred}

This section formally defines  inclusion-based points-to analysis
(Section~\ref{sec:pdpt}) and Dyck-reachability (Section~\ref{sec:pdcfl})
involved in our conditional hardness result.

\subsection{Inclusion-Based Points-to Analysis}\label{sec:pdpt}

Our work focuses on flow-insensitive inclusion-based (\myie, Andersen-style)
points-to analysis~\cite{SridharanF09the,HardekopfL07the}. The
control flow between  assignment statements in $P$ is irrelevant. Given a program $P$, a points-to analysis  determines the set of
variables that a pointer variable  might point to during program
execution. 

We consider a simple C-style language that contains the assignment statements of
the form ``$e_1 = e_2$'', where the expressions $e_1$
and $e_2$ are defined by the following context-free grammar:
\begin{align*}
e_1 &\rightarrow \texttt{\small ID} \mid \texttt{*}e_1\\
e_2 &\rightarrow \texttt{\small\&ID} \mid e_1. 
\end{align*}
The assignment statements permit arbitrary pointer dereferences (\myeg, \codeIn{****x
= **y}). Typically, all assignment statements in
the input program $P$ are normalized. 
The normalization procedure replaces the statements that involve multiple levels of
dereferencing by a sequence of statements that contain only one level of
dereferencing~\cite{Horwitz97precise}. After the normalization, each 
statement has one of the four forms: ``\codeIn{ID = \&ID}'', ``\codeIn{ID =
ID}'', ``\codeIn{ID = *ID}'' or ``\codeIn{*ID = ID}''.

Inclusion-based points-to analysis \textsc{PA}$\langle P, V\rangle$ is a set-constraint problem~\cite{FahndrichFSA98partial,HardekopfL07the}. 
It generates four types of inclusion constrains \mywrt the normalized
statements in program $P$. For each pointer variable $v \in V$, the goal of
points-to analysis is to compute a
\emph{points-to} set $\mathit{pt}(v)$ which contains all variables  that $v$ may
point to during execution. Table~\ref{tab:ptcons} gives the four constraints as well as their
corresponding statements and meanings. 

\begin{table}[t]
\begin{center}
\caption{Constraints for points-to analysis.}\label{tab:ptcons}
\begin{tabular}{ r l l l}
\hline
Type & Statement & Input Constraint & Inclusion  Constraint\\
\hline
\textsc{Address-of} & \codeIn{a = \&b} &  $\{b\} \subseteq a$ & $\mathit{loc}(b)
\in \mathit{pt}(a)$ \\
\textsc{Assignment} & \codeIn{a = b} &  $b \subseteq a$ & $\mathit{pt}(b)
\subseteq \mathit{pt}(a)$ \\
\textsc{Assign-star} & \codeIn{a = *b} & $*b \subseteq a$  & $\forall v \in \mathit{pt}(b):\mathit{pt}(v)
\subseteq \mathit{pt}(a)$ \\
\textsc{Star-assign} & \codeIn{*a = b} &  $b \subseteq *a$ & $\forall v \in \mathit{pt}(a):\mathit{pt}(b)
\subseteq \mathit{pt}(v)$ \\
\hline 
\end{tabular}

\end{center}
\end{table}

\begin{definition}[Inclusion-based points-to analysis]\label{def:pt}
Given a normalized program $P$ and a collection of input constraints based on Table~\ref{tab:ptcons},  the inclusion-based points-to analysis 
problem  is to solve the inclusion constraints and to determine if $p\in P$ can
point to $q \in P$ (\myie,
determine if $\mathit{loc}(q)
\in \mathit{pt}(p)$) for all $p, q \in P$.
\end{definition}

\begin{algorithm}[!t]
\footnotesize
\SetKwData{Null}{Null}
\SetKwData{oldIn}{oldIn}
\SetKwData{Up}{up}
\SetKwData{current}{current$_u$}
\SetKwData{previous}{previous$_u$}
\SetKwData{lastc}{last$_u$}
\SetKwData{lastp}{last$_p$}
\SetKwData{outaf}{\textsc{Out}$_{\alpha}$}
\SetKwData{outd}{out$_d$}
\SetKwData{outdb}{out$_{\bar{d}}$}
\SetKwData{rowu}{Row$_u$}

\SetKwFunction{Add}{Add}
\SetKwFunction{Init}{Init}
\SetKwFunction{PF}{ParititionFunc}
\SetKwInOut{Input}{Input}
\SetKwInOut{Output}{Output}
\DontPrintSemicolon


Let $G = (V, E)$ and initialize $V$ with program variables \\

\ForEach{constraint $\{b\} \subseteq a$}{
  $\mathit{pt}(a) \leftarrow \mathit{pt}(a) \cup \{b\}$
}
\ForEach{constraint $b \subseteq a$}{
  $E \leftarrow E \cup \{b\rightarrow a\}$
}
$W\leftarrow V$

\BlankLine
\While{$W \neq \emptyset$}{
$n \leftarrow \textsc{Select-from}(W)$\\
\ForEach{$v \in \mathit{pt}(n)$}{
  \ForEach{constraint $*n \subseteq a$}{
    \If{$v\rightarrow a \notin E$}{
      $E \leftarrow E \cup \{v \rightarrow a\}$  and    $W \leftarrow W \cup \{v\}$
    }

  }
  \ForEach{constraint $b \subseteq *n$}{
    \If{$b\rightarrow v \notin E$}{
      $E \leftarrow E \cup \{b \rightarrow v\}$  and    $W \leftarrow W \cup \{b\}$
    }

  }
}

\ForEach{$n\rightarrow z \in E$}{
  $\mathit{pt}(z) \leftarrow \mathit{pt}(z) \cup \mathit{pt}(n)$ \\
  \If{$\mathit{pt}(z)$ changed}{
    $W \leftarrow W \cup \{z\}$
  }
}

}

\caption{Inclusion-based points-to analysis algorithm.}\label{algo:pt}

\end{algorithm}

Algorithm~\ref{algo:pt} gives an algorithm for inclusion-based points-to
analysis based on computing dynamic transitive closure~\cite{HardekopfL07the}.

\begin{example}
Consider the following simple C-style program: \codeIn{a = $\&b$; b = $\&$d;} 
\codeIn{c = *a}.
We illustrate the Andersen-style pointer analysis by computing the $\mathit{pt}$ set 
for each variable.
According to the semantics introduced in Table \ref{tab:ptcons}, these three statements
represent $\mathit{loc}(b) \in \mathit{pt}(a), \mathit{loc}(d) \in \mathit{pt}(b)$ 
and $\forall v \in \mathit{pt}(a): \mathit{pt}(v) \subseteq \mathit{pt}(c)$.
After the fixed-point computation, we get the analysis result $\mathit{pt}(a) = \{ \mathit{loc}(b) \} ,$
$\mathit{pt}(b) = \{  \mathit{loc}(d)  \} $ and $\mathit{pt}(c) = \{ \mathit{loc}(d)   \} $.
\end{example}

\subsection{Dyck-Reachability}\label{sec:pdcfl}
Dyck-Reachability is a subclass of context-free language (CFL) reachability~\cite{Reps97program,ChatterjeeCP18opt,qirun13fast}.
A CFL-reachability problem instance 
contains a context free grammar $\mathit{CFG} = (\Sigma, N, P, S)$ and an edge
labeled digraph $G$. Each edge $u\xrightarrow{l}v \in G$ is labeled by a symbol $l =
\mathcal{L}(u, v) \in \Sigma \cup N$. Each path $p = v_0, v_1, v_2, \ldots, v_m$
in $G$ \emph{realizes} a string $\mathcal{R}(p)$ over the alphabet $\Sigma$ by concatenating the
edge labels in the path in order, \myie, $\mathcal{R}(p) = \mathcal{L}(v_0,
v_1)\mathcal{L}(v_1, v_2)\ldots\mathcal{L}(v_{m-1}, v_m)$. A path $p = v_0,
\ldots, v_m$ in $G$ is an \emph{$l$-path} if its realized string $\mathcal{R}(p)$
is either a terminal $l \in \Sigma$ or it can be derived
from a nonterminal $l \in N$. We represent an $l$-path from node $u$ to $v$ as a \emph{summary edge}
$u\xrightarrow{l} v$ in $G$. Moreover, we say node $v$ is \emph{$l$-reachable} from node
$u$ iff there exists a summary edge $u\xrightarrow{l} v$.

\begin{definition}[$L$-reachability]\label{def:cfgreach}
Given a grammar $\mathit{CFG}$ of a context-free language $L$ and an edge labeled
digraph $G$, the $L$-reachability problem is to compute all $S$-reachable nodes in $G$,
where $S$ is the start symbol in the $\mathit{CFG}$. 
\end{definition}


Dyck language $D_1$ is a context-free language that generates the strings of
one kind of properly matched parentheses. Formally, $D_1$ is specified by a
$\mathit{CFG} = (\Sigma, N, P, S)$ where $\Sigma = \{[_1, ]_1\}$, $N =
\{D_1, S\}$, and $S = \{D_1\}$. The production rules $P$ are \{$D_1 \rightarrow S$, $S \rightarrow [_1S]_1 \mid
SS\mid\epsilon$ \}.

\begin{definition}[$D_1$-reachability]
Given a grammar of Dyck language $D_1$ and an edge labeled
digraph $G$, the $D_1$-reachability problem is to compute all $D_1$-reachable nodes
 in $G$, where $D_1$ is the start
symbol in the grammar.
\end{definition}

\section{BMM-Hardness of $D_1$-Reachability and Points-to Analysis}\label{sec1:over}

We give two reductions to establish the conditional lower bounds of
$D_1$-reachability and inclusion-based points-to analysis, respectively.
In particular, the first reduction reduces Boolean Matrix Multiplication (BMM) to $D_1$-reachability and the second
reduction reduces  $D_1$-reachability to points-to analysis.
Our hardness results are conditioned on a widely believed BMM conjecture
about the complexity of multiplying two Boolean matrices as described in Conjecture~\ref{conj:bmm}.
Note that the cubic lower bound  of $D_k$-reachability is indeed known in the literature. For instance,
the work by \citet{HeintzeM97on} has proven a cubic lower bound  for a $D_2$-reachability data-flow
analysis problem conditioned on the hardness for solving \npda problems.
A recent result of \citet{ChatterjeeCP18opt}
proves the BMM-hardness of Dyck-reachability by giving a reduction from context-free
language (CFL) parsing which requires multiple kinds of parentheses.
Our hardness result on
$D_1$-reachability is slightly stronger since 
a lower bound on $D_1$-reachability implies a lower bound on
$D_k$-reachability, for any $k\geq 2$.

This section gives a high-level overview of our main results. In particular, 
Section~\ref{subsec:moti} gives a running example.
Sections~\ref{subsec:bmm} and ~\ref{subsec:d1} give the detailed description of the
two reductions, respectively.


\begin{figure}[t]
  \centering
  \begin{subfigure}[b]{0.4\linewidth}
\small
\begin{equation*}
\begin{spmatrix}{\mathbf{A}}
    0&1&0&0 \\
    0&0&0&0 \\
0&0&0&0\\
0&0&0&0
\end{spmatrix}
\begin{spmatrix}{\mathbf{B}}
    0&0&0&0 \\
    0&0&1&1 \\
0&0&0&0\\
0&0&0&0
\end{spmatrix}
=
\begin{spmatrix}{\mathbf{C}}
    0&0&1&1 \\
    0&0&0&0 \\
0&0&0&0\\
0&0&0&0
\end{spmatrix}
\end{equation*}

    \caption{A BMM problem instance.\label{fig:running1}}
  \end{subfigure}%
\hspace{2cm}
  \begin{subfigure}[b]{0.4\linewidth}

\begin{tikzpicture}[every path/.style={>=latex}]


  \node [style={draw,circle, inner sep=1pt, minimum size=0.6cm, scale=0.8}]
  (a) at (0,0)  {$x_0$ };
  \node [style={draw,circle, inner sep=1pt, minimum size=0.6cm, scale=0.8}]
  (a1) at (1.5,0)  {$x_1$ };
  \node [style={draw,circle, inner sep=1pt, minimum size=0.6cm, scale=0.8}]
  (a2) at (3,0)  {$x_2$ };
  \node [style={draw,circle, inner sep=1pt, minimum size=0.6cm, scale=0.8}]
  (a3) at (4.5,0)  {$x_3$ };
 
  \node [style={draw,circle, inner sep=1pt, minimum size=0.6cm, scale=0.8}]
  (b) at (0,-1.5)  {$y_0$ };
  \node [style={draw,circle, inner sep=1pt, minimum size=0.6cm, scale=0.8}]
  (b1) at (1.5,-1.5)  {$y_1$ };
  \node [style={draw,circle, inner sep=1pt, minimum size=0.6cm, scale=0.8}]
  (b2) at (3,-1.5)  {$y_2$ };
  \node [style={draw,circle, inner sep=1pt, minimum size=0.6cm, scale=0.8}]
  (b3) at (4.5,-1.5)  {$y_3$ };

  \node [style={draw,circle, inner sep=1pt, minimum size=0.6cm, scale=0.8}]
  (c) at (0,-3)  {$z_0$ };
  \node [style={draw,circle, inner sep=1pt, minimum size=0.6cm, scale=0.8}]
  (c1) at (1.5,-3)  {$z_1$ };
  \node [style={draw,circle, inner sep=1pt, minimum size=0.6cm, scale=0.8}]
  (c2) at (3,-3)  {$z_2$ };
  \node [style={draw,circle, inner sep=1pt, minimum size=0.6cm, scale=0.8}]
  (c3) at (4.5,-3)  {$z_3$ };

  \draw[->] (a) edge node[left, scale=0.8] {$[_1$}(b1);

  \draw[->] (b1) edge node[left, scale=0.8] {$]_1$}(c2);
  \draw[->] (b1) edge node[above, scale=0.8] {$]_1$}(c3);

\end{tikzpicture}

    \caption{A $D_1$-reachability problem instance.\label{fig:running2}}
  \end{subfigure}
\\
  \begin{subfigure}[b]{0.4\linewidth}

\begin{multicols}{2}
\begin{Verbatim}[fontsize=\small,commandchars=\\\{\},
  codes={\catcode`$=3\catcode`^=7},
  formatcom={\lstset{fancyvrb=false}}]

{\rmfamily\selectfont \textit{/* Encoding $[_1$ */}}
t$_1$ = &u$_0$; 
*v$_1$ = t$_1$; 

{\rmfamily\selectfont \textit{/* Encoding $]_1$ */}}
v$_1$ = &t$_2$; 
*t$_2$ = w$_2$; 

{\rmfamily\selectfont \textit{/* Encoding $]_1$ */}}
v$_1$ = &t$_3$; 
*t$_3$ = w$_3$; 

{\rmfamily\selectfont \textit{/* Mapping node. */}}
u$_0$ = *t$_4$;
t$_4$ = t$_5$;
t$_5$ = &t$_6$;
t$_6$ = &u$^\prime_0$    

{\rmfamily\selectfont \textit{/* Mapping node */}}
v$_1$ = *t$_7$;
t$_7$ = t$_8$;
t$_8$ = &t$_9$;
t$_9$ = &v$^\prime_1$    

{\rmfamily\selectfont \textit{/* Mapping node */}}
w$_2$ = *t$_{10}$;
t$_{10}$ = t$_{11}$;
t$_{11}$ = &t$_{12}$;
t$_{12}$ = &w$^\prime_2$    

{\rmfamily\selectfont \textit{/* Mapping node */}}
w$_3$ = *t$_{13}$;
t$_{13}$ = t$_{14}$;
t$_{14}$ = &t$_{15}$;
t$_{15}$ = &w$^\prime_3$    
\end{Verbatim}
\end{multicols}
    \caption{A C-style program $P$.\label{fig:running3}}
  \end{subfigure}

\caption{A running example.}\label{fig:running}
\end{figure}

\subsection{Overview}\label{subsec:moti}
Figure~\ref{fig:running} gives a concrete example to illustrate our 
reduction. 
We briefly describe the three problem instances in Figure~\ref{fig:running} as a gentle
introduction to our reduction. 

\begin{itemize}
\item \emph{BMM to $D1$-reachability.} Figure~\ref{fig:running1} gives a BMM instance with two $4\times4$
  matrices $\mathbf{A}$ and $\mathbf{B}$. Their product is given in matrix $\mathbf{C}$.
Figure~\ref{fig:running2}
depicts the transformed digraph $G$ as a $D_1$-reachability instance. 
Our reduction maps each non-zero $a_{ij} \in \mathbf{A}$ to $x_i \xrightarrow{[_1}
  y_j \in G$ and each non-zero $b_{ij} \in \mathbf{B}$ to $y_i \xrightarrow{]_1}
  z_j \in G$. Our reduction guarantees that non-zero $c_{ij} \in \mathbf{C} \Longleftrightarrow x_i \xrightarrow{D_1}
  z_j \in G$ (Theorem~\ref{thm:corr1}).

\item \emph{$D1$-reachability to points-to analysis.} 
Our second reduction takes as input the $D_1$-reachability problem instance in
Figure~\ref{fig:running2}. Note that 
the second reduction is general which does not rely on the output of the first
reduction. In Figure~\ref{fig:running}, we reuse the output graph $G$ in Figure~\ref{fig:running2}  for brevity. 
Figure~\ref{fig:running3} shows a normalized C-style program $P$ as a points-to
analysis problem instance. 
The program $P$ contains three sets of variables $\mathit{Var}_b$,
$\mathit{Var}_w$, and $\mathit{Var}_g$.
Our reduction maps each node $v \in G$ to a variable $v \in \mathit{Var}_b$. 
We also map each node $v\in G$ to a new address-taken variable $\&v^\prime \in
\mathit{Var}_g$. 
Moreover, we  construct program statements in $P$ to make variable $v$ point
to variable $v^\prime$, \myie, $\mathit{loc}(v^\prime) \in \mathit{pt}(v)$.
Set $\mathit{Var}_w$ contains auxiliary variables $t_i$ for
program statement construction.  
Our reduction guarantees that $u\xrightarrow{D_1} v \in G \Longleftrightarrow
\mathit{loc}(v^\prime) \in \mathit{pt}(u)$ in program $P$ (Theorem~\ref{thm:corr2}).
\end{itemize}

\subsection{Reducing BMM to $D_1$-Reachability}\label{subsec:bmm}
\begin{problem}[framed]{\large Reduction 1: From BMM to $D_1$-Reachability}
  Input: & Two $n\times n$ Boolean matrices $\mathbf{A}$ and $\mathbf{B}$; \\
  Output: & An edge-labeled digraph $G= (V, E)$, where $|V| = 3n$ and $|E|$ is
  equal to the number of non-zero entries in both $\mathbf{A}$ and $\mathbf{B}$. 
\end{problem}
\paragraph{Intuition.} A Boolean matrix is a matrix with entries from the set $\{0, 1\}$.
Given two $n \times n$ Boolean matrices $\mathbf{A}$ and $\mathbf{B}$, the BMM
problem is to compute the
product $\mathbf{C} = \mathbf{A}\times \mathbf{B}$, whose entries are defined by
$c_{ij} = \bigvee_{k=1}^{n} (a_{ik} \wedge b_{kj})$.
That is, $c_{ij} = 1$ if and only if there exists an integer  $k\in [0, n-1]$ such that
$a_{ik} = b_{kj} = 1$. Our basic reduction idea
is to treat every non-zero element $a_{ij} \in \mathbf{A}$ as a directed edge 
labeled by an open parenthesis ``$[_1$'' in
the output digraph $G$. Similarly, we treat every non-zero element
$b_{ij} \in \mathbf{B}$ as a ``$]_1$''-labeled edge in $G$. Our reduction
guarantees that  every non-zero
element $c_{ij} \in \mathbf{C}$ corresponds to a pair of nodes joined by a balanced-parenthesis path
(\myie, $D_1$-path) in
 graph $G$.

\begin{algorithm}[!t]
\footnotesize
\SetKwData{Null}{Null}
\SetKwData{oldIn}{oldIn}
\SetKwData{Up}{up}
\SetKwData{current}{current$_u$}
\SetKwData{previous}{previous$_u$}
\SetKwData{lastc}{last$_u$}
\SetKwData{lastp}{last$_p$}
\SetKwData{outaf}{\textsc{Out}$_{\alpha}$}
\SetKwData{outd}{out$_d$}
\SetKwData{outdb}{out$_{\bar{d}}$}
\SetKwData{rowu}{Row$_u$}

\SetKwFunction{Add}{Add}
\SetKwFunction{Init}{Init}
\SetKwFunction{PF}{ParititionFunc}
\SetKwInOut{Input}{Input}
\SetKwInOut{Output}{Output}
\DontPrintSemicolon

\Input{Two $n\times n$ Boolean matrices $\mathbf{A}$ and $\mathbf{B}$;  }
\Output{An edge-labeled graph $G=(V, E)$.}

\BlankLine
Introduce  nodes $x_i$, $y_i$, $z_i$ to $G$ where $i\in[0, n-1]$ \nllabel{algo:bmm1}\\

\ForEach{element $a_{ij} \in \mathbf{A}$ \nllabel{algo:bmm2}}{
  \If{$a_{ij}$ is $1$}{
    Insert edge $x_i \xrightarrow{[_1} y_j$ to $G$ \nllabel{algo:bmm4}\\ 
  }

}
\ForEach{element $b_{ij} \in \mathbf{B}$ }{
  \If{$b_{ij}$ is $1$}{
    Insert edge $y_i \xrightarrow{]_1} z_j$ to $G$ \nllabel{algo:bmm3}\\ 
  }
}

\caption{Reduction from Boolean matrix multiplication to $D_1$-reachability.}\label{algo:bmm2d1}

\end{algorithm}

\paragraph{Reduction.} Algorithm~\ref{algo:bmm2d1} gives the reduction
procedure. Given two $n\times n$ Boolean matrices $\mathbf{A}$ and $\mathbf{B}$, we introduce $3n$ nodes in
our graph (line~\ref{algo:bmm1}). We then insert edges based on the input matrices (lines~\ref{algo:bmm2}-\ref{algo:bmm3}).
Let $m$ be the number of non-zero entries in the input
matrices. Algorithm~\ref{algo:bmm2d1} outputs a digraph with $m$ edges.

\paragraph{Correctness.}
Let $V_x = \{x_0, \ldots,
x_{n-1}\}$, $V_y = \{y_0, \ldots, y_{n-1}\}$, and $V_z = \{z_0, \ldots, z_{n-1}\}$.
It is clear that Algorithm~\ref{algo:bmm2d1} is a linear-time
reduction. Specifically, given two $n \times n$ matrices, the reduction
generates a graph $G = (V, E)$ with $3n$ nodes where $V = V_x \cup
V_y \cup V_z$. 
To show $D_1$-reachability is BMM-hard, it suffices to prove the following
theorem on reduction correctness.

\begin{theorem}[Reduction Correctness]\label{thm:corr1}
Algorithm~\ref{algo:bmm2d1} is a linear-time reduction which takes as input two Boolean matrices $\mathbf{A}$ and
$\mathbf{B}$ and outputs a digraph $G = (V, E)$ where $V = V_x \cup V_y \cup
V_z$. Let $\mathbf{C} = \mathbf{A} \times \mathbf{B}$.
Element $c_{ij}\in \mathbf{C}$ is non-zero iff $z_j \in V_z$ is $D_1$-reachable from
$x_i \in V_x$ in $G$.
\end{theorem}

\begin{proof}
We  show that each non-zero element $c_{ij}$ corresponds to a
$D_1$-reachable pair in the constructed graph $G$ and vice versa.
\begin{itemize}
\item \emph{The ``$\Rightarrow$'' direction.} Based on the BMM definition, if
  $c_{ij} = 1$ there exists at least one $k$ such at $a_{ik} = b_{kj} =
  1$. Algorithm~\ref{algo:bmm2d1} inserts $x_i\xrightarrow{[_1} y_k$ (line~\ref{algo:bmm4}) and
    $y_k\xrightarrow{]_1} z_j$ (line~\ref{algo:bmm3}) to $G$. We have a path string $\mathcal{L}(x_i,
  y_k)\mathcal{L}(y_k, z_j) = [_1]_1 \in D_1$. Therefore, node $z_j$ is
  $D_1$-reachable from $x_i$.
\item \emph{The ``$\Leftarrow$'' direction.} Our constructed graph $G$ is a
  $3$-layered graph, \myie, graph $G$ contains three node sets $V_x$, $V_y$, $V_z$. 
All ``$[_1$''-labeled edges go from nodes in $V_x$ to nodes in $V_y$. All ``$]_1$''-labeled edges go from nodes in $V_y$ to nodes in $V_z$.
Therefore, the
  length of every
  $D_1$-path in $G$ is $2$. Specifically, every $D_1$-path begins with a node $x_i \in
  V_x$, passes through a node $y_k \in V_y$ and ends at a node $z_j \in
  V_z$. Based on Algorithm~\ref{algo:bmm2d1}, such a path corresponds to $a_{ik}
  = b_{kj} = 1$. Therefore, we have a non-zero element $c_{ij} \in \mathbf{C}$.
\end{itemize}
\end{proof}

\subsection{Reducing $D_1$-Reachability to Inclusion-Based Points-to Analysis}\label{subsec:d1}

\begin{problem}[framed]{{\large Reduction 2: From $D_1$-Reachability to Inclusion-Based
    Points-to Analysis}}
  Input: & An edge-labeled digraph $G= (V, E)$; \\
  Output: & A normalized C-style program with $4|V|+2|E|$ statements. 
\end{problem}

\paragraph{Intuition.} Our reduction takes as input a generic digraph $G = (V,
E)$ where each edge is labeled by either ``$[_1$'' or ``$]_1$''. 
Note that in $D_1$-reachability, every node is $D_1$-reachable from itself by an
empty path since
the $D_1$ nonterminal is nullable. However,  according to Definition~\ref{def:pt},
the points-to relation is not reflexive and we cannot generally
assume that variable \codeIn{p} points to its own location \codeIn{\&p}, \myie,
$\mathit{loc}(p) \notin \mathit{pt}(p)$ without any proper \textsc{Address-of}
statement in Table~\ref{tab:ptcons}.
To handle the reflexivity, we introduce a \emph{pointer variable set} $\mathit{Var}_b$
and a \emph{pointer address set} $\mathit{Var}_g$ in program $P$. 
For each node
$v \in V$, we construct a pointer variable $v \in \mathit{Var}_b$ and
$v^\prime\in \mathit{Var}_g$ in  program
$P$. In particular, we use $\mathit{loc}(v^\prime)$ to ``replace''
$\mathit{loc}(v)$ in our reduction
and insert additional program statements to make $\mathit{loc}(v^\prime) \in
\mathit{pt}(v)$ in program $P$.
We also introduce \emph{auxiliary nodes} $\mathit{Var}_w$ in $P$ to facilitate edge construction for $G$.
 Our key insight is
that the statements involving a dereferencing \codeIn{*a} and a referencing
\codeIn{\&a} exhibit a balanced-parentheses property. As a result, our reduction
constructs a C-style program $P$ with
those statements to express open- and close-parenthesis edges in $G$.
Finally, the reduction guarantees that each $D_1$-reachable pair $u\xrightarrow{D_1} v$ in
$G$ corresponds to the fact that $u \in \mathit{Var}_b$ points to $v^\prime \in
\mathit{Var}_g$ in program $P$, \myie,  $\mathit{loc}(v^\prime) \in \mathit{pt}(u)$.

\begin{algorithm}[!t]
\footnotesize
\SetKwData{Null}{Null}
\SetKwData{oldIn}{oldIn}
\SetKwData{Up}{up}
\SetKwData{current}{current$_u$}
\SetKwData{previous}{previous$_u$}
\SetKwData{lastc}{last$_u$}
\SetKwData{lastp}{last$_p$}
\SetKwData{outaf}{\textsc{Out}$_{\alpha}$}
\SetKwData{outd}{out$_d$}
\SetKwData{outdb}{out$_{\bar{d}}$}
\SetKwData{rowu}{Row$_u$}

\SetKwFunction{Add}{Add}
\SetKwFunction{Init}{Init}
\SetKwFunction{PF}{ParititionFunc}
\SetKwInOut{Input}{Input}
\SetKwInOut{Output}{Output}
\DontPrintSemicolon

\Input{An edge-labeled digraph $G = (V, E)$;  }
\Output{A C-style program $P$.}

\BlankLine

Introduce $|E|+3|V|$ temporary variables $t_i \in \mathit{Var}_w$ to program $P$ \\
$i \leftarrow 1$\\

\ForEach{node $v \in V$ \nllabel{algo:pt1}}{
  \tcp{construct $\mathit{Var}_b$ and $\mathit{Var}_g$.}
  Introduce variables $v \in \mathit{Var}_b$ and $v^\prime \in \mathit{Var}_g$ 
  to program $P$ \nllabel{algo:nodemap}\\

  \tcp{make $v$ point to $v^\prime$, \myie, $\mathit{loc}(v^\prime) \in \mathit{pt}(v)$.}
  Insert a statement ``$\codeInM{v = *t_i;}$'' to $P$ \nllabel{algo:pte1} \\
  Insert a statement ``$\codeInM{t_i = t_{i+1};}$'' to $P$  \\
  Insert a statement ``$\codeInM{t_{i+1} = \&t_{i+2};}$'' to $P$  \\
  Insert a statement ``$\codeInM{t_{i+2} = \&v^\prime;}$'' to $P$  \nllabel{algo:pte2}\\
  $i\leftarrow i+3$ \nllabel{algo:pt2}
}


\ForEach{edge $(u, v) \in E$\nllabel{algo:pt3}}{
  \If{edge label $\mathcal{L}(u,v)$ is $[_1$}{
    \tcp{encode open parentheses.}
    Insert a statement ``$\codeInM{t_{i} = \&u;}$'' to $P$ \nllabel{algo:d21} \\
    Insert a statement ``$\codeInM{*v = t_i;}$'' to $P$ \nllabel{algo:d22} \\
    $i\leftarrow i+1$
  }\Else{
      \tcp{encode close parentheses.}
    Insert a statement ``$\codeInM{u = \&t_i;}$'' to $P$ \nllabel{algo:d31} \\
    Insert a statement ``$\codeInM{*t_i = y;}$'' to $P$ \nllabel{algo:d32} \\
    $i\leftarrow i+1$ \nllabel{algo:pt4}
  }

}

\caption{Reduction from $D_1$-reachability to inclusion-based points-to analysis.}\label{algo:d1topt}

\end{algorithm}

\paragraph{Reduction.}  Algorithm~\ref{algo:d1topt} gives the reduction from
$D_1$-reachability to inclusion-based points-to
analysis. 
It inserts statements to program $P$ based on the nodes (lines~\ref{algo:pt1}-\ref{algo:pt2}) and edges (lines~\ref{algo:pt3}-\ref{algo:pt4}) in the input
graph $G$.
Line~\ref{algo:nodemap} maps every node $v \in V$ to a unique element in 
$\mathit{Var}_b$ and $\mathit{Var}_g$, respectively.
Therefore, the mapping between any two of sets $V$, $\mathit{Var}_b$ and
$\mathit{Var}_g$ is a bijection.

\paragraph{Correctness.} Algorithm~\ref{algo:d1topt} is a linear-time
reduction. Specifically, given a graph $G = (V, E)$ with $|V|$ nodes and  $|E|$ edges, our reduction outputs
a normalized C-style program with $5|V|+|E|$ variables and $4|V|+2|E|$ statements.
Algorithm~\ref{algo:d1topt} maps every node $v\in V$ in the input graph to a
unique variable $v\in \mathit{Var}_b$ and a unique variable $v^\prime \in \mathit{Var}_g$ in
the output program.
Based on Theorem~\ref{thm:bmm1}, to show that points-to analysis is BMM-hard, it suffices to prove the following
theorem on reduction correctness.
\begin{theorem}[Reduction Correctness]\label{thm:corr2}
Algorithm~\ref{algo:d1topt} takes as input a digraph $G=(V, E)$ and outputs a
C-style program $P$ with $O(E)$ variables and $O(E)$ statements. All nodes $v\in V$
are represented as variables $\codeInM{v}$ and $\codeInM{v^\prime}$ in $P$.
Node $v$ is $D_1$-reachable from node $u$ in $G$ iff $\mathit{loc}(v^\prime) \in
\mathit{pt}(u)$ (\myie, variable $\codeInM{u}$ points to variable
$\codeInM{v^\prime}$) in program $P$ based on solving the inclusion constraints given in
Table~\ref{tab:ptcons}.
\end{theorem}
We discuss  reduction correctness
in Sections~\ref{sec1:pt1} and~\ref{sec1:pt2}.

\section{Inclusion-Based Points-to Analysis via $\mathit{Pt}$-Reachability}\label{sec1:pt1}
This section describes a graphical representation of any normalized C-style program $P$. We denote the
representation as a \emph{pointer expression graph} (PEG) $G_P$.
Each PEG node
 corresponds to a pointer expression $e \in \{\text{\codeIn{\&a}, \codeIn{a},
   \codeIn{*a}}\}$ in $P$, where \codeIn{a} is a pointer variable.
Each PEG edge $(u, v)$ is labeled by a letter $\mathcal{L}(u, v) \in \{r,
\overline{r}, s, \overline{s}, as, \overline{as}, sa, \overline{sa}, d, \overline{d}\}$.
Specifically, PEGs contain
 two kinds of edges defined as follows.

\begin{itemize}
\item \emph{Program Edge:} For each \textsc{Address-of} statement ``\codeIn{a =
  \&b}'' in $P$, we insert an edge $\codeInM{a}\xrightarrow{r}\codeInM{\&b}$ in $G_P$. For each \textsc{Assignment} statement ``\codeIn{a =
  b}'' in $P$, we insert an edge $\codeInM{a}\xrightarrow{s}\codeInM{b}$ in $G_P$.  For each \textsc{Assign-star} statement ``\codeIn{a =
  *b}'' in $P$, we insert an edge $\codeInM{a}\xrightarrow{\mathit{as}}\codeInM{*b}$ in $G_P$.  For each \textsc{Star-assign} statement ``\codeIn{*a =
  b}'' in $P$, we insert an edge $\codeInM{*a}\xrightarrow{\mathit{sa}}\codeInM{b}$ in $G_P$.   
\item \emph{Dereference Edge:} For each variable $\codeInM{a} \in P$, we insert
  two edges
  $\codeInM{\&a}\xrightarrow{d}\codeInM{a}$ and
  $\codeInM{a}\xrightarrow{d}\codeInM{*a}$ in $G_P$. 
\end{itemize}

PEGs are bidirected. Let $t$ be an edge label. For each edge $u\xrightarrow{t}v$, there  always
exists a
reverse edge $v\xrightarrow{\overline{t}}u$ in $G_P$. Similarly,  for each edge $u\xrightarrow{\overline{t}}v$ , there  always
exists a
reverse edge $v\xrightarrow{t}u$. 
Based on Table~\ref{tab:ptcons}, there are four types of program statement in a
normalized program $P$.
In PEG $G_P$, each program edge  corresponds to a statement in
$P$. The dereference edges respect the pointer semantics described in
Section~\ref{sec:pdpt}. Thus, it is straightforward to see that the
mapping between $G_P$ and $P$ is bijective.
\begin{lemma}\label{lem:peg}
The pointer expression graph $G_P$ is equivalent to program $P$.
\end{lemma}

Next, we describe the \textsc{PA} constraint resolution rules  given in
Table~\ref{tab:ptcons} using CFL-reachability.
We define a context-free grammar $\mathit{Pt} = (\Sigma_g, N_g, P_g, S_g)$ for points-to
analysis. Specifically, the alphabet $\Sigma_g = \{d,\overline{d}, r,\overline{r},
s,\overline{s}, \mathit{as},\overline{as}, \mathit{sa},\overline{sa}\}$ contains
the edge labels in PGE $G_P$. The grammar contains three nonterminals with a start
symbol $\mathit{Pt}$, \myie, $N_g =
\{S,\overline{S}, \mathit{Pt}, \overline{Pt}\}$ and $S_g = \{\mathit{Pt}\}$.

Table~\ref{tab:ptgraph} describes four types of constraints. In particular, the ``Meaning'' column in Table~\ref{tab:ptgraph} shows that the four constrains can be expressed in terms
of the points-to and subset constraints.
Let  $\mathit{Var}$ and $\mathit{Addr}$ be the sets of pointer variables and
variable addresses in program $P$, respectively. 
The points-to constraint defines a binary relation $\textsc{Pt}_{\in} \in
\mathit{Var}\times\mathit{Addr}$ and the subset constraint defines a binary
relation $\textsc{Subset}_{\subseteq} \in \mathit{Var}\times\mathit{Var}$. For
instance, $(p, \codeInM{\&q}) \in \textsc{Pt}_{\in}$ corresponds to the
points-to constraint $\mathit{loc}(q) \in \mathit{pt}(p)$ in Table~\ref{tab:ptcons}, \myie,  \codeIn{p} points to \codeIn{q}.
Similarly, $(p, q) \in \textsc{Subset}_{\subseteq}$ corresponds to the
subset constraint $\mathit{pt}(q) \subseteq \mathit{pt}(p)$ in
Table~\ref{tab:ptcons}, \myie,  \codeIn{q}'s points-to set is a subset of \codeIn{p}'s.

\begin{table}[t]
\begin{center}
\caption{Constraints for points-to analysis via graph reachability.}\label{tab:ptgraph}
\begin{tabular}{ r l c l}
\hline
Type & Statement & Constraint & Meaning\\
\hline
\adjustbox{valign=c}{\textsc{Address-of}} & \adjustbox{valign=c}{\codeIn{a = \&b}} &
\adjustbox{valign=c}{\begin{tikzpicture}
        \node [style={circle, inner sep=1pt,scale=0.8}]           (a) at (0,0)  { \codeIn{a} };
        \node [style={circle, inner sep=1pt,scale=0.8}]           (b) at
        (2,0)  { \codeIn{\&b} };
\begin{scope}[>=stealth]
   \draw[->] (a) edge node[above, midway,scale=0.6] {$r$}  (b);
   \draw[dashed,->] (a) edge[bend right=-30] node[above, midway,scale=0.6] {$\mathit{Pt}$}  (b);
\end{scope}
    \end{tikzpicture}} &\multirow{2}{*}{\adjustbox{valign=c}{$\begin{aligned} 
\mathit{Pt} \rightarrow &~S~r~\mid~r\\
\mathit{S} \rightarrow &~S~S\\
\mathit{S} \rightarrow &~s
\end{aligned}$
}} \\
\textsc{Assignment} & \codeIn{a = b} &
\begin{tikzpicture}
        \node [style={circle, inner sep=1pt,scale=0.8}]           (a) at (0,0)  { \codeIn{a} };
        \node [style={circle, inner sep=1pt,scale=0.8}]           (b) at
        (2,0)  { \codeIn{b} };
\begin{scope}[>=stealth]
   \draw[->] (a) edge node[above, midway,scale=0.6] {$s$}  (b);
   \draw[dashed,->] (a) edge[bend right=-30] node[above, midway,scale=0.6] {$S$}  (b);
\end{scope}
    \end{tikzpicture} &  \\
\textsc{Assign-star} & \codeIn{a = *b} &
\adjustbox{valign=c}{\begin{tikzpicture}
        \node [style={circle, inner sep=1pt,scale=0.8}]           (a) at (0,0)  { \codeIn{a} };
        \node [style={circle, inner sep=1pt,scale=0.8}]           (b) at
        (1,0)  { \codeIn{*b} };
        \node [style={circle, inner sep=1pt,scale=0.8}]           (b1) at
        (1,-1)  { \codeIn{b} };
        \node [style={circle, inner sep=1pt,scale=0.8}]           (c) at
        (2,0)  { \codeIn{v} };
        \node [style={circle, inner sep=1pt,scale=0.8}]           (c1) at
        (2,-1)  { \codeIn{\&v} };
\begin{scope}[>=stealth]
   \draw[->] (a) edge node[below, midway,scale=0.6] {$\mathit{as}$}  (b);
   \draw[->] (b1) edge node[left, midway,scale=0.6] {$d$}  (b);
   \draw[->] (c1) edge node[right, midway,scale=0.6] {$d$}  (c);
   \draw[->] (b1) edge node[above, midway,scale=0.6] {$\mathit{Pt}$}  (c1);
   \draw[dashed,->] (a) edge[bend right=-25] node[above, midway,scale=0.6] {$S$}  (c);

\end{scope}
    \end{tikzpicture}}& \adjustbox{valign=c}{$\begin{aligned} 
\mathit{S} \rightarrow &~\mathit{as}~\overline{d}~\mathit{Pt}~d
\end{aligned}$} \\
\textsc{Star-assign} & \codeIn{*a = b} & 
\adjustbox{valign=c}{\begin{tikzpicture}
        \node [style={circle, inner sep=1pt,scale=0.8}]           (a) at (2,0)  { \codeIn{b} };
        \node [style={circle, inner sep=1pt,scale=0.8}]           (b) at
        (0,0)  { \codeIn{*a} };
        \node [style={circle, inner sep=1pt,scale=0.8}]           (b1) at
        (0,-1)  { \codeIn{a} };
        \node [style={circle, inner sep=1pt,scale=0.8}]           (c) at
        (1,0)  { \codeIn{v} };
        \node [style={circle, inner sep=1pt,scale=0.8}]           (c1) at
        (1,-1)  { \codeIn{\&v} };
\begin{scope}[>=stealth]
   \draw[dashed,->] (c) edge node[below, midway,scale=0.6] {$S$}  (a);
   \draw[->] (b1) edge node[left, midway,scale=0.6] {$d$}  (b);
   \draw[->] (c1) edge node[right, midway,scale=0.6] {$d$}  (c);
   \draw[->] (b1) edge node[above, midway,scale=0.6] {$\mathit{Pt}$}  (c1);
   \draw[->] (b) edge[bend right=-25] node[above, midway,scale=0.6] {$\mathit{sa}$}  (a);

\end{scope}
    \end{tikzpicture}} & \adjustbox{valign=c}{$\begin{aligned} 
\mathit{S} \rightarrow &~\overline{d}~\overline{\mathit{Pt}}~d~\mathit{sa}
\end{aligned}$} \\
\hline 
\end{tabular}

\end{center}
\end{table}

In our $Pt$-reachability formulation, we use two nonterminals $\mathit{Pt}$ and
$S$ to express the two relations $\textsc{Pt}_{\in}$ and
$\textsc{Subset}_{\subseteq}$, respectively.
That is,  $(p, \codeInM{\&q}) \in \textsc{Pt}_{\in}$ iff
$p\xrightarrow{\mathit{Pt}}\codeInM{\&q}$ in PEGs and 
$(p, q) \in \textsc{Subset}_{\subseteq}$ iff $p\xrightarrow{S}q$ in PEGs.
Nonterminal $\overline{S}$ and $\overline{Pt}$ denote the inverses of $S$ and
$\mathit{Pt}$, respectively.
Next, we discuss the semantics of the four types of constrains in
Table~\ref{tab:ptgraph} and express them using a CFG $\mathit{Pt}$.

\begin{itemize}
\item \textsc{Assignment}: The assignment statement \codeIn{a = b} in program $P$ is represented as an edge $a\xrightarrow{s}
b$ edge in the PEG. 
The \textsc{Assignment} constraint means the \codeIn{b}'s points-to set
  is a subset of \codeIn{a}'s.  The resolution in Table~\ref{tab:ptcons} adds an
  inclusion constraint between $\mathit{pt}(b)$ and $\mathit{pt}(a)$, \myie,
  $(a, b) \in \textsc{Subset}_\subseteq$. 
Therefore, we have a summary edge $a\xrightarrow{S}b$ in PEG.
Recall that the new summary is generated based on edge $a\xrightarrow{s}
b$, we encode it using a rule $S\rightarrow s$. Since inclusion
  constraints are always transitive and reflexive, we describe this
  using a rule $S\rightarrow S~S \mid \epsilon$.
\item \textsc{Address-of}: The address-of statement \codeIn{a = \&b} in program $P$ is represented as an edge $a\xrightarrow{r}
\&b$ edge in the PEG. 
The \textsc{Address-of} constraint means variable \codeIn{b} belongs to
  \codeIn{a}'s points-to set. The resolution in Table~\ref{tab:ptcons} assigns \codeIn{b}'s location
  $\mathit{loc}(b)$ to \codeIn{a}'s points-to set $\mathit{pt}(a)$, \myie,   $(a, \codeInM{\&b}) \in \textsc{Pt}_\in$.
Therefore, we have a summary edge $a\xrightarrow{\mathit{Pt}}\codeInM{\&b}$
based on the PEG edge $a\xrightarrow{\mathit{r}}\codeInM{\&b}$.
Note that \codeIn{\&b} should also belong to
  all supersets of $\mathit{pt}(a)$, \myie, $(c, \codeInM{\&b}) \in
  \textsc{Pt}_\in$ for all c such that  $(c, a) \in \textsc{Subset}_\subseteq$. 
In PEG, we insert a summary edge $c\xrightarrow{\mathit{Pt}}\codeInM{\&b}$ based
on $c\xrightarrow{\mathit{S}}a$ and $a\xrightarrow{\mathit{r}}\codeInM{\&b}$.
Since nonterminal $S$ is nullable, we can combine these cases by describing them using a rule
  $\mathit{Pt} \rightarrow S~r$.
\item \textsc{Assign-star}: The assign-star statement \codeIn{a = *b} in program $P$ is represented as an edge $a\xrightarrow{\mathit{as}}
*b$ in the PEG. 
The \textsc{Assign-star} constraint means that, for all variables
  \codeIn{v} that \codeIn{b} points to (\myie, for all $v$ such that $(b, \codeInM{\&v}) \in
  \textsc{Pt}_\in$), \codeIn{v}'s points-to set should be a
  subset of \codeIn{a}'s (\myie, $(a, v) \in \textsc{Subset}_\subseteq$). 
The relation $(b, \codeInM{\&v}) \in
  \textsc{Pt}_\in$ is expressed as a summary edge $b\xrightarrow{\mathit{Pt}}
  \codeInM{\&v}$. 
The newly generated relation $(a, v) \in
  \textsc{Subset}_\subseteq$ can be described as a summary edge $a\xrightarrow{\mathit{S}} v$ in the PEG.
To sum up, we insert a summary edge $a\xrightarrow{\mathit{S}} v$ based on $a\xrightarrow{\mathit{as}}
*b$ and $b\xrightarrow{\mathit{Pt}}
  \codeInM{\&v}$.
In the PEG, there are two deference edges $*b\xrightarrow{\overline{d}}
b$ and $\codeInM{\&v}\xrightarrow{d}v$ that bridge the gaps.
Note that the deference edges always exist among pointer expressions \mywrt the
pointer semantics.
Therefore, in the PEG,  we generate a new summary edge
$\codeInM{a}\xrightarrow{S}\codeInM{v}$ based on four summary edges
$\codeInM{a}\xrightarrow{\mathit{as}}\codeInM{*b}\xrightarrow{\overline{d}}\codeInM{b}\xrightarrow{\mathit{Pt}}\codeInM{\&v}\xrightarrow{d}\codeInM{v}$. Finally,
we describe it
using a rule   $\mathit{S} \rightarrow \mathit{as}~\overline{d}~\mathit{Pt}~d$.
\item \textsc{Star-assign}: 
The star-assign statement \codeIn{*a = b} in program $P$ is represented as an edge $*a\xrightarrow{\mathit{sa}}
b$ in the PEG. 
The \textsc{Star-assign} constraint means that,
for all variables
  \codeIn{v} that \codeIn{a} points to (\myie, for all $v$ such that $(a, \codeInM{\&v}) \in
  \textsc{Pt}_\in$), \codeIn{b}'s points-to set should be a
  subset of \codeIn{v}'s (\myie, $(v, b) \in \textsc{Subset}_\subseteq$).  
The relation $(a, \codeInM{\&v}) \in
  \textsc{Pt}_\in$ is expressed as a summary edge $a\xrightarrow{\mathit{Pt}}
  \codeInM{\&v}$. 
The newly generated relation $(v, b) \in
  \textsc{Subset}_\subseteq$ can be described as a summary edge $v\xrightarrow{\mathit{S}} b$ in the PEG.
To sum up, we insert a summary edge $v\xrightarrow{\mathit{S}} b$ based on a
reverse edge $\codeInM{\&v}\xrightarrow{\overline{Pt}}a$ and $*a\xrightarrow{\mathit{sa}}
b$.
Like the \textsc{Assign-star} case, there are two deference edges $v\xrightarrow{\overline{d}}
\&v$ and $a\xrightarrow{d}*a$ in the PEG.
Therefore,
we generate a new summary edge
$\codeInM{v}\xrightarrow{S}\codeInM{b}$ based on four summary edges
$\codeInM{v}\xrightarrow{\overline{d}}\codeInM{\&v}\xrightarrow{\overline{\mathit{Pt}}}\codeInM{a}\xrightarrow{d}\codeInM{*a}\xrightarrow{\mathit{sa}}\codeInM{b}$. We describe it
using a rule   $\mathit{S} \rightarrow \overline{d}~\overline{\mathit{Pt}}~d~\mathit{sa}$.
\end{itemize}

The context-free grammar described in Table~\ref{tab:ptgraph} fully captures the
constraint resolution in Table~\ref{tab:ptcons}. Specifically, each terminal 
summary edge depicts either a program statement in $P$ or a pointer dereference
(\myie, initial constraint) and each nonterminal summary edge
describes a new constraint generated during constraint resolution. And all
constraints are encoded using summary edges.

To distinguish itself from the other grammars used in our discussion, we rename the
$\mathit{Pt}$ nonterminal in Table~\ref{tab:ptgraph} to $\mathit{PA}$. We give the full
productions in Figure~\ref{fig:gpa1} by expanding all inverse nonterminals. Combined with
Lemma~\ref{lem:peg}, we have the following equivalence result.

\begin{lemma}\label{lem:peg2pt}
Given a program $P$ and its representative PEG, computing all-pairs
$\mathit{Pt}$-reachability in PEG is equivalent to
computing inclusion-based pointer analysis of $P$. Specifically, node
\codeIn{\&b} is $\mathit{Pt}$-reachable from node \codeIn{a} in PEG iff
$\mathit{loc}(b) \in \mathit{pt}(a)$.
\end{lemma}

\subsection{PEG $G_P$ Generated in Reduction}
To facilitate the reduction from $D_1$-reachability, the C-style program $P$ generated in Section~\ref{subsec:d1}
has some structural properties. Those properties  play a pivotal role to show the
correctness described in Theorem~\ref{thm:corr2}.

Unless otherwise noted, we refer to $G_P$ as the PEG that corresponds to the
output program $P$ of
Algorithm~\ref{algo:d1topt} in our discussion.
Algorithm~\ref{algo:d1topt} partitions 
the variables in the output program $P$ into three disjoint sets
$\mathit{Var}_b$, $\mathit{Var}_g$, and $\mathit{Var}_w$. Specifically, 
each node $v \in G$ in
the $D_1$-reachability instance becomes two variables $v \in \mathit{Var}_b$ and
$v^\prime \in \mathit{Var}_g$ (line~\ref{algo:nodemap}). 
Each variable $t \in \mathit{Var}_w$ is an
auxiliary variable which has been used in either node-processing
(lines~\ref{algo:pt1}-\ref{algo:pt2}) or in edge-processing
(line~\ref{algo:pt3}-\ref{algo:pt4}) of Algorithm~\ref{algo:d1topt}.
We further partition $\mathit{Var}_w$ into two disjoint sets $\mathit{Var}_{w1}$
and $\mathit{Var}_{w2}$ based on node-processing and edge-processing, respectively.
In the constructed PEG $G_P$, we associate each node with a color and a shape:
\begin{itemize}
\item A \emph{white node} (\myie, \tikz\draw[black,fill=white] (0,0) circle (.5ex); or $\square$) represents a variable $t \in \mathit{Var}_{w1} \cup
  \mathit{Var}_{w2}$. White nodes correspond to  auxiliary variables used in $P$.
\item A \emph{black node} (\myie, $\blacksquare$) represents a variable $v \in \mathit{Var}_{b}$. Black
  nodes correspond to the graph nodes in $D_1$-reachability. Black nodes 
  appear in the points-to query.
\item A \emph{gray node} (\myie, \tikz\draw[black,fill=gray,opacity=0.8] (0,0) circle (.5ex);) represents a variable $ v^\prime \in \mathit{Var}_{g}$. Gray
  nodes correspond to the graph nodes in $D_1$-reachability. Gray nodes 
  appear as the address-taken variables in the points-to query.
\item A \emph{square node} (\myie, $\square$ or $\blacksquare$) represents a variable $v \in \mathit{Var}_{b} \cup
  \mathit{Var}_{w2} $. Those nodes are constructed to model the graph
  edges in  $D_1$-reachability shown in Figure~\ref{fig:g2peg}.
\item A \emph{circle node} (\myie, \tikz\draw[black,fill=white] (0,0) circle (.5ex); or \tikz\draw[black,fill=gray,opacity=0.8] (0,0) circle (.5ex);) represents a variable $v \in \mathit{Var}_{g} \cup
  \mathit{Var}_{w1}$. Those nodes are constructed to model the graph
  nodes in  $D_1$-reachability based on Figure~\ref{fig:g2peg}.
\end{itemize}

\begin{lemma}\label{pty:node}
Based on Figure~\ref{fig:g2peg}, Algorithm~\ref{algo:d1topt} maps each node $v
\in G$ to a black square and a gray circle node in PEG $G_P$. Both mappings are bijective.
\end{lemma}

\begin{figure}[t]
\begin{center}
\small
\begin{tabular}{ c|c| c |c}
\hline
Construction Type &Input Graph $G$   & PEG $G_P$ & Program $P$\\
\hline
\multirow{2}{*}[-1cm]{\small \textsc{Edge-with-sa}} &
\adjustbox{valign=c}{\begin{tikzpicture}

        \node [style={circle,draw, inner sep=1pt,scale=0.85}]           (a) at (0,0)  { $x$ };
        \node [style={circle,draw, inner sep=1pt,scale=0.8}]           (c) at (1,0)  { $y$ };

   \draw[->] (a) -- node[above, midway,scale=0.8] {$[_1$}  (c);

    \end{tikzpicture}}    &

\adjustbox{valign=c}{\begin{tikzpicture}

        \node [style={rectangle,inner sep=.8pt}]           (a) at (0,0)  { $\blacksquare$};
        \node [style={rectangle,inner sep=.8pt}]           (b) at (0,-.8)  { $\square$};
        \node [style={rectangle,inner sep=.8pt}]           (c) at (0.8,-0.8)  { $\square$};
        \node [style={rectangle,inner sep=.8pt}]           (d) at (1.6,-0.8)  { $\square$};
        \node [style={rectangle,inner sep=.8pt}]           (e) at (1.6,-1.6)  { $\blacksquare$};

        \node [style={rectangle,inner sep=.8pt}, left=2pt of a,scale=.8]           (x)   {
          $x$};
        \node [style={rectangle,inner sep=.8pt}, right=2pt of e,scale=.8]           (y)   {
          $y$};
        \node [style={rectangle,inner sep=.8pt}]           (t) at (-0.8,0) {};
        \node [style={rectangle,inner sep=.8pt},above=10pt of a]           (t) {};
        \node [style={rectangle,inner sep=.8pt},below=10pt of e]           (t)  {};

\begin{scope}
   \draw[->] (a) -- node[left, midway,scale=0.7] {$\overline{d}$}  (b);
   \draw[->] (b) -- node[above, midway,scale=0.7] {$\overline{r}$}  (c);
   \draw[->] (c) -- node[above, midway,scale=0.7] {$\overline{sa}$}  (d);
   \draw[->] (d) -- node[right, midway,scale=0.7] {$\overline{d}$}  (e);
\end{scope}
    \end{tikzpicture}}
& \shortstack{\codeIn{t$_i$ = \&x;} \\ \codeIn{*y = t$_i$;}}
 \\
\cline{2-4}
 &\adjustbox{valign=c}{\begin{tikzpicture}

        \node [style={circle,draw, inner sep=1pt,scale=0.85}]           (a) at (0,0)  { $x$ };
        \node [style={circle,draw, inner sep=1pt,scale=0.8}]           (c) at (1,0)  { $y$ };

   \draw[->] (a) -- node[above, midway,scale=0.8] {$]_1$}  (c);

    \end{tikzpicture}}    &

\adjustbox{valign=c}{\begin{tikzpicture}

        \node [style={rectangle,inner sep=.8pt}]           (a) at (0,0)  { $\blacksquare$};
        \node [style={rectangle,inner sep=.8pt}]           (b) at (0.8,0)  { $\square$};
        \node [style={rectangle,inner sep=.8pt}]           (c) at (0.8,0.8)  { $\square$};
        \node [style={rectangle,inner sep=.8pt}]           (d) at (0.8,1.6)  { $\square$};
        \node [style={rectangle,inner sep=.8pt}]           (e) at (1.6,1.6)  { $\blacksquare$};

        \node [style={rectangle,inner sep=.8pt}, left=2pt of a,scale=.8]           (x)   {
          $x$};
        \node [style={rectangle,inner sep=.8pt}, right=2pt of e,scale=.8]           (y)   {
          $y$};
        \node [style={rectangle,inner sep=.8pt},]           (t) at (-0.8,0) {};
        \node [style={rectangle,inner sep=.8pt},above=10pt of e]           (t)
              {};
        \node [style={rectangle,inner sep=.8pt},below=10pt of a]           (t)  {};

\begin{scope}
   \draw[->] (a) -- node[above, midway,scale=0.7] {$r$}  (b);
   \draw[->] (b) -- node[right, midway,scale=0.7] {$d$}  (c);
   \draw[->] (c) -- node[right, midway,scale=0.7] {$d$}  (d);
   \draw[->] (d) -- node[above, midway,scale=0.7] {$\mathit{sa}$}  (e);
\end{scope}
    \end{tikzpicture}}
 & \shortstack{\codeIn{x = \&t$_i$;} \\ \codeIn{*t$_i$ = y;}} \\
\hline

{\small\textsc{Node-with-as-s}}& \adjustbox{valign=c}{\begin{tikzpicture}
        \node [style={circle, inner sep=1pt,scale=0.8}]           (c) at (0,0)  { };
        \node [style={circle,draw, inner sep=1pt,scale=0.85}]           (a) at (0.5,0)  { $x$ };
        \node [style={circle, inner sep=1pt,scale=0.8}]           (c) at (1,0)  {  };


    \end{tikzpicture}}    &

\adjustbox{valign=c}{\begin{tikzpicture}
\definecolor{red}{gray}{0.7}
        \node [style={rectangle,inner sep=.8pt}]           (a) at (0,0)  { $\blacksquare$};
        \node [style={rectangle,inner sep=.8pt}]           (b) at (0.8,0)  {
          $\circ$};
        \node [style={rectangle,inner sep=.8pt}]           (b1) at (0.8,-.8)  { $\circ$};
        \node [style={rectangle,inner sep=.8pt}]           (c) at (1.6,-0.8)  { $\circ$};
        \node [style={rectangle,inner sep=.8pt}]           (d) at (2.4,-0.8)  {
          $\circ$};
        \node [style={rectangle,inner sep=.8pt}]           (d1) at (2.4,0)  { $\circ$};
        \node [style={circle,draw, inner sep=2pt, fill=red}]           (e) at (3.2,0)  { };

        \node [style={rectangle,inner sep=.8pt}, below=2pt of a,scale=.8]           (x)   {
          $x$};
        \node [style={rectangle,inner sep=.8pt}, below=2pt of e,scale=.8]           (y)   {
          $\&x^\prime$};
        \node [style={rectangle,inner sep=.8pt},above=10pt of a]           (t) {};
        \node [style={rectangle,inner sep=.8pt},below=10pt of c]           (t)  {};

\begin{scope}
   \draw[->] (a) -- node[above, midway,scale=0.7] {$\mathit{as}$}  (b);
   \draw[->] (b) -- node[left, midway,scale=0.7] {$\overline{d}$}  (b1);
   \draw[->] (b1) -- node[above, midway,scale=0.7] {$s$}  (c);
   \draw[->] (c) -- node[above, midway,scale=0.7] {$r$}  (d); 
   \draw[->] (d) -- node[right, midway,scale=0.7] {$d$}  (d1); 
   \draw[->] (d1) -- node[above, midway,scale=0.7] {$r$}  (e);
\end{scope}
    \end{tikzpicture}}
& \shortstack{\codeIn{x = *t$_i$;} \\ \codeIn{t$_i$ = t$_{i+1}$;}
  \\ \codeIn{t$_{i+1}$ = \&t$_{i+2}$;} \\ \codeIn{t$_{i+2}$ = \&x$^\prime$;}} \\

\hline 
\end{tabular}
\caption{Edge construction for PEG $G_P$.} \label{fig:g2peg}
\end{center}
\end{figure}

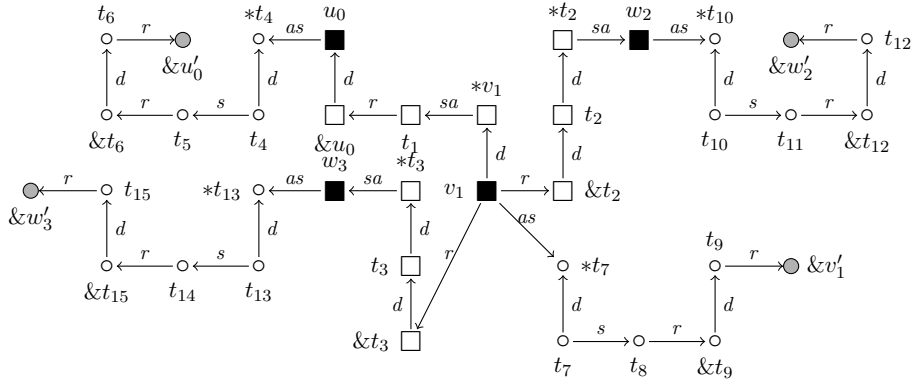
\begin{figure}[t]
\centering
\begin{tikzpicture}
\definecolor{red}{gray}{0.7}
        \node [style={rectangle,inner sep=.8pt}]           (v1) at (0,-1)  { $\blacksquare$};
        \node [style={rectangle,inner sep=.8pt}, left=2pt of v1,scale=.8]           (v1desc)   {$v_1$};

        \node [style={rectangle,inner sep=.8pt}]           (starv1) at (0,0)  { $\square$};
        \node [style={rectangle,inner sep=.8pt}, above=2pt of starv1,scale=.8]           (starv1desc)   {$*v_1$};

        \node [style={rectangle,inner sep=.8pt}]           (t1) at (-1, 0)  { $\square$};
        \node [style={rectangle,inner sep=.8pt}, below=2pt of t1,scale=.8]           (t1desc)   {$t_1$};

        \node [style={rectangle,inner sep=.8pt}]           (adru0) at (-2,0)  { $\square$};
        \node [style={rectangle,inner sep=.8pt}, below=2pt of adru0,scale=.8]           (adru0desc)   {$\&u_0$};

        \node [style={rectangle,inner sep=.8pt}]           (u0) at (-2, 1)  { $\blacksquare$};
        \node [style={rectangle,inner sep=.8pt}, above=2pt of u0,scale=.8]           (u0desc)   {$u_0$};

        \node [style={rectangle,inner sep=.8pt}]           (start4) at (-3, 1)  { $\circ$};
        \node [style={rectangle,inner sep=.8pt}, above=2pt of start4,scale=.8]           (start4desc)   {$*t_4$};

        \node [style={rectangle,inner sep=.8pt}]           (t4) at (-3,0)  { $\circ$};
        \node [style={rectangle,inner sep=.8pt}, below=2pt of t4,scale=.8]           (t4desc)   {$t_4$};

        \node [style={rectangle,inner sep=.8pt}]           (t5) at (-4, 0)  { $\circ$};
        \node [style={rectangle,inner sep=.8pt}, below=2pt of t5,scale=.8]           (t5desc)   {$t_5$};

        \node [style={rectangle,inner sep=.8pt}]           (adrt6) at (-5,0)  { $\circ$};
        \node [style={rectangle,inner sep=.8pt}, below=2pt of adrt6,scale=.8]           (adrt6desc)   {$\&t_6$};

        \node [style={rectangle,inner sep=.8pt}]           (t6) at (-5,1)  { $\circ$};
        \node [style={rectangle,inner sep=.8pt}, above=2pt of t6,scale=.8]           (t6desc)   {$t_6$};

        \node [style={circle,draw, inner sep=2pt, fill=red}]           (adruprime0) at (-4,1)  {};
        \node [style={rectangle,inner sep=.8pt}, below=2pt of adruprime0,scale=.8]           (adruprime0desc)   {$\&u^\prime_0$};

        \node [style={rectangle,inner sep=.8pt}]           (adrt2) at (1, -1)  { $\square$};
        \node [style={rectangle,inner sep=.8pt}, right=2pt of adrt2,scale=.8]           (adrt2desc)   {$\&t_2$};

        \node [style={rectangle,inner sep=.8pt}]           (t2) at (1, 0)  { $\square$};
        \node [style={rectangle,inner sep=.8pt}, right=2pt of t2,scale=.8]           (t2desc)   {$t_2$};

        \node [style={rectangle,inner sep=.8pt}]           (start2) at (1,1)  { $\square$};
        \node [style={rectangle,inner sep=.8pt}, above=2pt of start2,scale=.8]           (start2desc)   {$*t_2$};

        \node [style={rectangle,inner sep=.8pt}]           (w2) at (2,1)  { $\blacksquare$};
        \node [style={rectangle,inner sep=.8pt}, above=2pt of w2,scale=.8]           (w2desc)   {$w_2$};

        \node [style={rectangle,inner sep=.8pt}]           (start10) at (3,1)  { $\circ$};
        \node [style={rectangle,inner sep=.8pt}, above=2pt of start10,scale=.8]           (start10desc)   {$*t_{10}$};

        \node [style={rectangle,inner sep=.8pt}]           (t10) at (3,0)  { $\circ$};
        \node [style={rectangle,inner sep=.8pt}, below=2pt of t10,scale=.8]           (t10desc)   {$t_{10}$};

        \node [style={rectangle,inner sep=.8pt}]           (t11) at (4, 0)  { $\circ$};
        \node [style={rectangle,inner sep=.8pt}, below=2pt of t11,scale=.8]           (t11desc)   {$t_{11}$};

        \node [style={rectangle,inner sep=.8pt}]           (adrt12) at (5, 0)  { $\circ$};
        \node [style={rectangle,inner sep=.8pt}, below=2pt of adrt12,scale=.8]           (adrt12desc)   {$\&t_{12}$};

        \node [style={rectangle,inner sep=.8pt}]           (t12) at (5,1)  { $\circ$};
        \node [style={rectangle,inner sep=.8pt}, right=2pt of t12,scale=.8]           (t12desc)   {$t_{12}$};

        \node [style={circle,draw, inner sep=2pt, fill=red}]           (adrwprime2) at (4,1)  {};
        \node [style={rectangle,inner sep=.8pt}, below=2pt of adrwprime2,scale=.8]           (adrwprime2desc)   {$\&w^\prime_2$};

        \node [style={rectangle,inner sep=.8pt}]           (adrt3) at (-1,-3)  { $\square$};
        \node [style={rectangle,inner sep=.8pt}, left=2pt of adrt3,scale=.8]           (adrt3desc)   {$\&t_{3}$};

        \node [style={rectangle,inner sep=.8pt}]           (t3) at (-1,-2)  { $\square$};
        \node [style={rectangle,inner sep=.8pt}, left=2pt of t3,scale=.8]           (t3desc)   {$t_{3}$};

        \node [style={rectangle,inner sep=.8pt}]           (start3) at (-1,-1)  { $\square$};
        \node [style={rectangle,inner sep=.8pt}, above=2pt of start3,scale=.8]           (start3desc)   {$*t_{3}$};

        \node [style={rectangle,inner sep=.8pt}]           (w3) at (-2,-1)  { $\blacksquare$};
        \node [style={rectangle,inner sep=.8pt}, above=2pt of w3,scale=.8]           (w3desc)   {$w_3$};

        \node [style={rectangle,inner sep=.8pt}]           (start13) at (-3,-1)  { $\circ$};
        \node [style={rectangle,inner sep=.8pt}, left=2pt of start13,scale=.8]           (start13desc)   {$*t_{13}$};

        \node [style={rectangle,inner sep=.8pt}]           (t13) at (-3,-2)  { $\circ$};
        \node [style={rectangle,inner sep=.8pt}, below=2pt of t13,scale=.8]           (t13desc)   {$t_{13}$};

        \node [style={rectangle,inner sep=.8pt}]           (t14) at (-4,-2)  { $\circ$};
        \node [style={rectangle,inner sep=.8pt}, below=2pt of t14,scale=.8]           (t14desc)   {$t_{14}$};

        \node [style={rectangle,inner sep=.8pt}]           (adrt15) at (-5,-2)  { $\circ$};
        \node [style={rectangle,inner sep=.8pt}, below=2pt of adrt15,scale=.8]           (adrt15desc)   {$\&t_{15}$};

        \node [style={rectangle,inner sep=.8pt}]           (t15) at (-5,-1)  { $\circ$};
        \node [style={rectangle,inner sep=.8pt}, right=2pt of t15,scale=.8]           (t15desc)   {$t_{15}$};

        \node [style={circle,draw, inner sep=2pt, fill=red}]           (adrwprime3) at (-6, -1)  {};
        \node [style={rectangle,inner sep=.8pt}, below=2pt of adrwprime3,scale=.8]           (adrwprime3desc)   {$\&w^\prime_3$};

        \node [style={rectangle,inner sep=.8pt}]           (start7) at (1,-2)  { $\circ$};
        \node [style={rectangle,inner sep=.8pt}, right=2pt of start7,scale=.8]           (start7desc)   {$*t_{7}$};

        \node [style={rectangle,inner sep=.8pt}]           (t7) at (1,-3)  { $\circ$};
        \node [style={rectangle,inner sep=.8pt}, below=2pt of t7,scale=.8]           (t7desc)   {$t_{7}$};

        \node [style={rectangle,inner sep=.8pt}]           (t8) at (2,-3)  { $\circ$};
        \node [style={rectangle,inner sep=.8pt}, below=2pt of t8,scale=.8]           (t8desc)   {$t_{8}$};

        \node [style={rectangle,inner sep=.8pt}]           (adrt9) at (3,-3)  { $\circ$};
        \node [style={rectangle,inner sep=.8pt}, below=2pt of adrt9,scale=.8]           (adrt9desc)   {$\&t_{9}$};

        \node [style={rectangle,inner sep=.8pt}]           (t9) at (3,-2)  { $\circ$};
        \node [style={rectangle,inner sep=.8pt}, above=2pt of t9,scale=.8]           (t9desc)   {$t_{9}$};

        \node [style={circle,draw, inner sep=2pt, fill=red}]           (adrvprime1) at (4,-2)  {};
        \node [style={rectangle,inner sep=.8pt}, right=2pt of adrvprime1,scale=.8]           (adrvprime1desc)   {$\&v^\prime_1$};


        \draw[->] (v1) -- node[right, midway,scale=0.7] {$\mathit{d}$}  (starv1);
        \draw[->] (starv1) -- node[above, midway,scale=0.7] {$\mathit{sa}$}  (t1);
        \draw[->] (t1) -- node[above, midway,scale=0.7] {$\mathit{r}$}  (adru0);
        \draw[->] (adru0) -- node[right, midway,scale=0.7] {$\mathit{d}$}  (u0);
        \draw[->] (u0) -- node[above, midway,scale=0.7] {$\mathit{as}$}  (start4);
        \draw[->] (t4) -- node[right, midway,scale=0.7] {$\mathit{d}$}  (start4);
        \draw[->] (t4) -- node[above, midway,scale=0.7] {$\mathit{s}$}  (t5);
        \draw[->] (t5) -- node[above, midway,scale=0.7] {$\mathit{r}$}  (adrt6);
        \draw[->] (adrt6) -- node[right, midway,scale=0.7] {$\mathit{d}$}  (t6);
        \draw[->] (t6) -- node[above, midway,scale=0.7] {$\mathit{r}$}  (adruprime0);
        \draw[->] (v1) -- node[above, midway,scale=0.7] {$\mathit{r}$}  (adrt2);
        \draw[->] (adrt2) -- node[right, midway,scale=0.7] {$\mathit{d}$}  (t2);
        \draw[->] (t2) -- node[right, midway,scale=0.7] {$\mathit{d}$}  (start2);
        \draw[->] (start2) -- node[above, midway,scale=0.7] {$\mathit{sa}$}  (w2);
        \draw[->] (w2) -- node[above, midway,scale=0.7] {$\mathit{as}$}  (start10);
        \draw[->] (t10) -- node[right, midway,scale=0.7] {$\mathit{d}$}  (start10);
        \draw[->] (t10) -- node[above, midway,scale=0.7] {$\mathit{s}$}  (t11);
        \draw[->] (t11) -- node[above, midway,scale=0.7] {$\mathit{r}$}  (adrt12);
        \draw[->] (adrt12) -- node[right, midway,scale=0.7] {$\mathit{d}$}  (t12);
        \draw[->] (t12) -- node[above, midway,scale=0.7] {$\mathit{r}$}  (adrwprime2);
        \draw[->] (v1) -- node[above, midway,scale=0.7] {$\mathit{r}$}  (adrt3);
        \draw[->] (adrt3) -- node[left, midway,scale=0.7] {$\mathit{d}$}  (t3);
        \draw[->] (t3) -- node[right, midway,scale=0.7] {$\mathit{d}$}  (start3);
        \draw[->] (start3) -- node[above, midway,scale=0.7] {$\mathit{sa}$}  (w3);
        \draw[->] (t13) -- node[right, midway,scale=0.7] {$\mathit{d}$}  (start13);
        \draw[->] (w3) -- node[above, midway,scale=0.7] {$\mathit{as}$}  (start13);
        \draw[->] (t13) -- node[above, midway,scale=0.7] {$\mathit{s}$}  (t14);
        \draw[->] (t14) -- node[above, midway,scale=0.7] {$\mathit{r}$}  (adrt15);
        \draw[->] (adrt15) -- node[right, midway,scale=0.7] {$\mathit{d}$}  (t15);
        \draw[->] (t15) -- node[above, midway,scale=0.7] {$\mathit{r}$}  (adrwprime3);
        \draw[->] (v1) -- node[above, midway,scale=0.7] {$\mathit{as}$}  (start7);
        \draw[->] (t7) -- node[right, midway,scale=0.7] {$\mathit{d}$}  (start7);
        \draw[->] (t7) -- node[above, midway,scale=0.7] {$\mathit{s}$}  (t8);
        \draw[->] (t8) -- node[above, midway,scale=0.7] {$\mathit{r}$}  (adrt9);
        \draw[->] (adrt9) -- node[right, midway,scale=0.7] {$\mathit{d}$}  (t9);
        \draw[->] (t9) -- node[above, midway,scale=0.7] {$\mathit{r}$}  (adrvprime1);

\end{tikzpicture}
    
\caption{PEG $G_P$ of program $P$ in Figure~\ref{fig:running3}.}\label{picture:pegforrunning3}
\end{figure}

\paragraph{Language $\mathit{Pt}$.}To see the connection between $D_1$ and $\mathit{Pt}$,
we expand the $\mathit{Pt}$ rules in Table~\ref{tab:ptgraph}. In particular, we keep the
start nonterminal $\mathit{Pt}$ and replace any other occurrence of
$\mathit{Pt}$ with $S~r$. Figure~\ref{fig:gpa1} gives the rewritten grammar.
Since nonterminal $S$ is nullable, the rewritten grammar in
Figure~\ref{fig:gpa1} is equivalent to the
original grammar in Table~\ref{tab:ptgraph}.

\begin{example}
Figure~\ref{picture:pegforrunning3} gives the generated PEG for the program in Figure
~\ref{fig:running3}.
In the graph, we can see that there is a $\mathit{Pt}$-path from $w_3$ to $\&w^\prime_3$.
The realized string of path $w_3\rightarrow\ast t_{13} \rightarrow
t_{13}\rightarrow t_{14}\rightarrow \&t_{15}\rightarrow t_{15}\rightarrow \&w^\prime_3$ 
is  ``$as\ \overline{d}\ s\ r\ d\ r$''. 
According to the production rules (\ref{rule:bb0}) and (\ref{rule:bb1}), 
the realized string belongs to the $\mathit{Pt}$ language. 
And node $\&w^\prime_3$ is  $\mathit{Pt}$-reachable from $w_3$ in $G_P$.
We further observe that in the original program (Figure~\ref{fig:running3}), 
the last four lines of code are related to $w_3$ and $w^\prime_3$.
According to Table \ref{tab:ptcons}, we have  constraints
$\forall v \in \mathit{pt}(t_{13}): \mathit{pt}(v) \subseteq \mathit{pt}(w_3)$,
$\mathit{pt}(t_{14}) \subseteq \mathit{pt}(t_{13})$,
$\mathit{loc}(t_{15}) \in \mathit{pt}(t_{14})$ and 
$\mathit{loc}(w^\prime_3) \in \mathit{pt}(t_{15})$. 
Therefore, we have $\mathit{loc}(w^\prime_3) \in \mathit{pt}(w_3)$. 
Finally, we have $w_3\xrightarrow{\mathit{Pt}}\&w_3^\prime \in G_P
\Leftrightarrow \mathit{loc}(w^\prime_3) \in \mathit{pt}(w_3) \in P$ (Lemma~\ref{lem:peg2pt}).

Consider another pair of nodes $u_0$ and $w_3$. The path between them realizes
the string ``$\overline{d}\ \overline{r}\ \overline{sa}\ \overline{d}\ r\ d\ d\ sa$''.
Based on the production rule
(\ref{rule:bb3}),
(\ref{rule:bb4}) and
(\ref{rule:bb7}),
we can see that it can be derived by nonterminal $S$ in Figure~\ref{fig:gpa1}.
On the other hand, we can  extract the set constraints  from
the original program, the statements \codeIn{$t_1$ = \&$u_0$; *$v_1$ = $t_1$;}
\codeIn{$v_1 = \&t_3$; *$t_3$ = $w_3$} yield the relation
$\mathit{pt}(w_3) \subseteq \mathit{pt}(u_0)$. The $S$-reachability in PEG $G_P$
and set constraint resolution in $P$  agree
on the subset relation.

Finally, we consider the pair of nodes $w_3$ and $w_2$. The path between these two nodes
realizes the string ``$\overline{sa}\ \overline{d}\ \overline{d}\ \overline{r} \ r\ d\ d\ sa$''.
This word cannot be recognized by the $S$ language or the $Pt$ language. 
From the program statements, the set constraints can not establish a
subset relation or a points-to relation between the two corresponding variables. 
\end{example}

\section{$D_1$-Reachability and $\mathit{Pt}$-Reachability}\label{sec1:pt2}
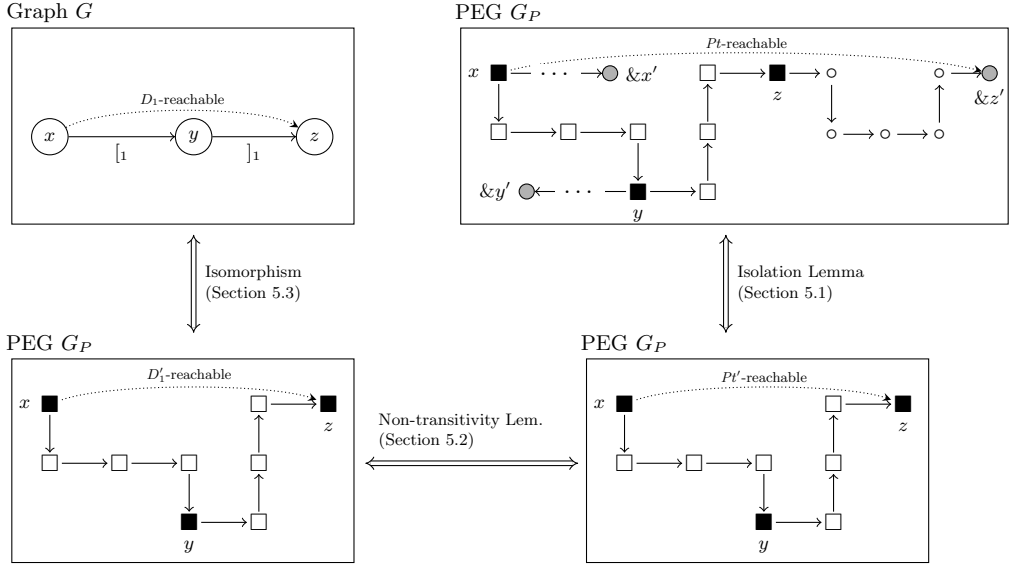
\begin{figure}[t]
\footnotesize
\centering
\newcommand{\myw}{0.6}
\newcommand{\myww}{0.4}
\newcommand{\myh}{0.5}
\begin{tikzpicture}[mynode/.style={rectangle,fill=white,anchor=center}]]
\definecolor{red}{gray}{0.7}
        \node [style={rectangle,inner sep=.8pt}]           (a) at (0,0)  { $\blacksquare$};
        \node [style={rectangle,inner sep=.8pt}, below = \myh of a]           (b)   { $\square$};
        \node [style={rectangle,inner sep=.8pt}, right = \myw of b]
        (c)   { $\square$};
        \node [style={rectangle,inner sep=.8pt}, right = \myw of c]
        (d)   { $\square$};
        \node [style={rectangle,inner sep=.8pt}, below = \myh of d]
        (e)   { $\blacksquare$};
        \node [style={rectangle,inner sep=.8pt}, right = \myw of e]
        (f)   { $\square$};
        \node [style={rectangle,inner sep=.8pt}, above = \myh of f]
        (g)   { $\square$};
        \node [style={rectangle,inner sep=.8pt}, above = \myh of g]
        (h)   { $\square$};
        \node [style={rectangle,inner sep=.8pt}, right = \myw of h]
        (i)   { $\blacksquare$};

        \node [style={rectangle,inner sep=.8pt}, left=2pt of a,scale=.8]           (x)   {
          $x$};
        \node [style={rectangle,inner sep=.8pt}, below=2pt of e,scale=.8]           (y)   {
          $y$};
        \node [style={rectangle,inner sep=.8pt},below=2pt of i,scale=.8]
        (z) {$z$};


     \node [style={rectangle,inner sep=.8pt}, right = 7.4cm of a]           (a1) at (0,0)  { $\blacksquare$};
        \node [style={rectangle,inner sep=.8pt}, below = \myh of a1]           (b1)   { $\square$};
        \node [style={rectangle,inner sep=.8pt}, right = \myw of b1]
        (c1)   { $\square$};
        \node [style={rectangle,inner sep=.8pt}, right = \myw of c1]
        (d1)   { $\square$};
        \node [style={rectangle,inner sep=.8pt}, below = \myh of d1]
        (e1)   { $\blacksquare$};
        \node [style={rectangle,inner sep=.8pt}, right = \myw of e1]
        (f1)   { $\square$};
        \node [style={rectangle,inner sep=.8pt}, above = \myh of f1]
        (g1)   { $\square$};
        \node [style={rectangle,inner sep=.8pt}, above = \myh of g1]
        (h1)   { $\square$};
        \node [style={rectangle,inner sep=.8pt}, right = \myw of h1]
        (i1)   { $\blacksquare$};

        \node [style={rectangle,inner sep=.8pt}, left=2pt of a1,scale=.8]           (x1)   {
          $x$};
        \node [style={rectangle,inner sep=.8pt}, below=2pt of e1,scale=.8]           (y1)   {
          $y$};
        \node [style={rectangle,inner sep=.8pt},below=2pt of i1,scale=.8]
        (z1) {$z$};

  \node [style={circle, inner sep=1pt, minimum size=0.6cm, scale=0.8}, above=4cm of a]
  (ax)  { };
    \node [style={draw,circle, inner sep=1pt, minimum size=0.6cm, scale=0.8}, below=0.35cm of ax]
  (axx)  {$x$ };
    \node [style={draw,circle, inner sep=1pt, minimum size=0.6cm, scale=0.8}, right=1.4 of axx]
  (ex)  {$y$ };
    \node [style={draw,circle, inner sep=1pt, minimum size=0.6cm, scale=0.8}, right=1.1 of ex]
  (ix)  {$z$ };
  
        \node [style={rectangle,inner sep=.8pt}, right = 5.5cm of ax]           (a2)   { $\blacksquare$};
        \node [style={rectangle,inner sep=.8pt}, below = \myh of a2]           (b2)   { $\square$};
        \node [style={rectangle,inner sep=.8pt}, right = \myw of b2]
        (c2)   { $\square$};
        \node [style={rectangle,inner sep=.8pt}, right = \myw of c2]
        (d2)   { $\square$};
        \node [style={rectangle,inner sep=.8pt}, below = \myh of d2]
        (e2)   { $\blacksquare$};
        \node [style={rectangle,inner sep=.8pt}, right = \myw of e2]
        (f2)   { $\square$};
        \node [style={rectangle,inner sep=.8pt}, above = \myh of f2]
        (g2)   { $\square$};
        \node [style={rectangle,inner sep=.8pt}, above = \myh of g2]
        (h2)   { $\square$};
        \node [style={rectangle,inner sep=.8pt}, right = \myw of h2]
        (i2)   { $\blacksquare$};
        \node [style={circle, inner sep=.8pt}, right = \myww of i2]           (j2)
              {$\circ$};
        \node [style={circle, inner sep=.8pt}, below = \myh of j2]           (k2)
              {$\circ$}; 
        \node [style={circle, inner sep=.8pt}, right = \myww of k2]           (l2)
              {$\circ$}; 
        \node [style={circle, inner sep=.8pt}, right = \myww of l2]           (m2)
              {$\circ$}; 
        \node [style={circle, inner sep=.8pt}, above = \myh of m2]           (n2)
              {$\circ$}; 
        \node [style={circle,draw, inner sep=2pt, fill=red}, right = \myww of n2]           (o2)   { };

        \node [style={rectangle,inner sep=.8pt}, left=2pt of a2,scale=.8]           (x)   {
          $x$};
        \node [style={rectangle,inner sep=.8pt}, below=2pt of e2,scale=.8]           (y)   {
          $y$};
        \node [style={rectangle,inner sep=.8pt},below=2pt of i2,scale=.8]
        (z) {$z$};
        \node [style={rectangle,inner sep=.8pt},below=2pt of o2,scale=.8]           (zz) {$\&z^\prime$};
        \node [style={circle,draw, inner sep=2pt, fill=red}, right = 1.2 of a2]           (ox2)   { }; 
        \node [style={circle,draw, inner sep=2pt, fill=red}, left = 1.2 of e2]           (oy2)   { };
        \node [style={rectangle,inner sep=.8pt},right=2pt of ox2,scale=.8]           (zzz) {$\&x^\prime$};
        \node [style={rectangle,inner sep=.8pt},left=2pt of oy2,scale=.8]           (zzzz) {$\&y^\prime$};

\begin{scope}
\draw[->] (axx) edge node[below, scale=0.8] {$[_1$}(ex);
\draw[->] (ex) edge node[below, scale=0.8] {$]_1$}(ix);
   \draw [->,>=stealth,out=30,in=150,looseness=0.5,densely dotted] (axx) to node[above,midway, scale=0.6]{$D_1$-reachable}  (ix);
   \draw[->] (a) -- node[left, midway,scale=0.6] {}  (b);
   \draw[->] (b) -- node[above, midway,scale=0.6] {}  (c);
   \draw[->] (c) -- node[above, midway,scale=0.6] {}  (d);
   \draw[->] (d) -- node[right, midway,scale=0.6] {}  (e);
   \draw[->] (e) -- node[below, midway,scale=0.6] {}  (f);
   \draw[->] (f) -- node[right, midway,scale=0.6] {}  (g);
   \draw[->] (g) -- node[right, midway,scale=0.6] {}  (h);
   \draw[->] (h) -- node[above, midway,scale=0.6] {}  (i);
  
   \draw [->,>=stealth,out=20,in=160,looseness=0.5,densely dotted] (a) to node[above,midway, scale=0.6]{$D_1^\prime$-reachable}  (i);
   
 \draw[->] (a1) -- node[left, midway,scale=0.6] {}  (b1);
   \draw[->] (b1) -- node[above, midway,scale=0.6] {}  (c1);
   \draw[->] (c1) -- node[above, midway,scale=0.6] {}  (d1);
   \draw[->] (d1) -- node[right, midway,scale=0.6] {}  (e1);
   \draw[->] (e1) -- node[below, midway,scale=0.6] {}  (f1);
   \draw[->] (f1) -- node[right, midway,scale=0.6] {}  (g1);
   \draw[->] (g1) -- node[right, midway,scale=0.6] {}  (h1);
   \draw[->] (h1) -- node[above, midway,scale=0.6] {}  (i1);
  
   \draw [->,>=stealth,out=20,in=160,looseness=0.5,densely dotted] (a1) to node[above,midway, scale=0.6]{$\mathit{Pt}^\prime$-reachable}  (i1);
   \draw[->] (a2) -- node[left, midway,scale=0.6] {}  (b2);
   \draw[->] (b2) -- node[above, midway,scale=0.6] {}  (c2);
   \draw[->] (c2) -- node[above, midway,scale=0.6] {}  (d2);
   \draw[->] (d2) -- node[right, midway,scale=0.6] {}  (e2);
   \draw[->] (e2) -- node[below, midway,scale=0.6] {}  (f2);
   \draw[->] (f2) -- node[right, midway,scale=0.6] {}  (g2);
   \draw[->] (g2) -- node[right, midway,scale=0.6] {}  (h2);
   \draw[->] (h2) -- node[above, midway,scale=0.6] {}  (i2);
   \draw[->] (i2) -- node[above, midway,scale=0.6] {}  (j2);
   \draw[->] (j2) -- node[left, midway,scale=0.6] {}  (k2);
   \draw[->] (k2) -- node[below, midway,scale=0.6] {}  (l2);
   \draw[->] (l2) -- node[below, midway,scale=0.6] {}  (m2);
   \draw[->] (m2) -- node[right, midway,scale=0.6] {}  (n2);
   \draw[->] (n2) -- node[below, midway,scale=0.6] {}  (o2);
   \draw [->,>=stealth,out=15,in=165,looseness=0.5,densely dotted] (a2) to node[above,midway, scale=0.6]{$\mathit{Pt}$-reachable}  (o2);
    \draw[->] (a2) --node[mynode,pos=.5] {$\dots$} (ox2);
   \draw[->] (e2) --node[mynode,pos=.5] {$\dots$} (oy2); 
   
\end{scope}
\draw[draw=black] ($ (a) + (-0.5,0.6) $) rectangle ++(4.5,-2.7);
\draw[draw=black] ($ (a1) + (-0.5,0.6) $) rectangle ++(4.5,-2.7);
\draw[draw=black] ($ (ax) + (-0.5,0.6) $) rectangle ++(4.5,-2.6);
\draw[draw=black] ($ (a2) + (-0.5,0.6) $) rectangle ++(7.2,-2.6);

 \node [style={rectangle,inner sep=.8pt}, scale=.9]           (tt1) at ($ (a) + (0,.8) $) {PEG $G_P$};  
 \node [style={rectangle,inner sep=.8pt}, scale=.9]           (tt1) at ($ (a1) + (0,.8) $) {PEG $G_P$};  
 \node [style={rectangle,inner sep=.8pt}, scale=.9]           (tt1) at ($ (a2) + (0,.8) $) {PEG $G_P$};
 \node [style={rectangle,inner sep=.8pt}, scale=.9]           (tt1) at ($ (ax) + (0,.8) $) {Graph $G$};
\draw[implies-implies, double equal sign distance] ($ (ex) + (0,-1.3) $) --
node[rectangle, above=.0cm, right=2pt,scale=0.7,text width=2.5cm]{Isomorphism  (Section~\ref{subsec:c})} ($ (ex) + (0,-2.6) $);
\draw[implies-implies, double equal sign distance] ($ (ex) + (7,-1.3) $) -- node[rectangle, above=.0cm, right=2pt,scale=0.7,text width=3cm]{Isolation Lemma (Section~\ref{subsec:a})} ($ (ex) + (7,-2.6) $);
\draw[implies-implies, double equal sign distance] ($ (g) + (1.4,0) $) -- node[rectangle, above=3pt, scale=0.7, text width=3.5cm]{Non-transitivity Lem. (Section~\ref{subsec:b})} ($ (g) + (4.2,0) $);
\end{tikzpicture}

\caption{Overview of reductions.}
\label{figure:reduction_overview}
    
\end{figure}

In this section, we prove that $D_1$-reachability in $G$ is equivalent to
$\mathit{Pt}$-reachability among black and gray nodes in $G_P$.

Our basic idea is to simplify the language $\mathit{Pt}$ and convert it to a
$D_1$-like language called $D_1^\prime$. 
Figure~\ref{figure:reduction_overview} gives an overview of our reduction.
In particular, in $G_P$, we prove that $\mathit{Pt}$-reachability among black
and gray nodes is equivalent to a simplified $\mathit{Pt}^\prime$-reachability
with only black nodes. We further simplify $\mathit{Pt}^\prime$-reachability and
convert it to $D_1^\prime$-reachability in $G_P$. Finally, we show that the
$D^\prime_1$-reachability problem in $G_P$ is equivalent to the
$D_1$-reachability problem in $G$.

Note that the simplifications mentioned above do not hold for general
PEG. However, our reduction in Algorithm~\ref{algo:d1topt} emits a specialized
PEG. Figure~\ref{fig:gpa} shows all grammars involved in our proof. Our key insight is to leverage the properties in the constructed PEG for
simplifying the $\mathit{Pt}$-reachability problem. In particular, our
constructed PEG introduces two aspects of restrictions:
\begin{itemize}
\item The nodes in the PEG $G_P$ are of different shapes (\myie, square and circle
  nodes). Figure~\ref{figure:reduction_overview} gives an illustration of our constructed PEG. Our
  construction in Figure~\ref{fig:g2peg} guarantees that the reachability
  among square nodes does not involve circle nodes. This helps us eliminate a
  few rules in grammar
  $\mathit{Pt}$ and obtain a simpler grammar $\mathit{Pt}^\prime$ (Section~\ref{subsec:a}).
\item The nodes in the PEG $G_P$ are of different colors (\myie, black, white,
  and gray) as well. The color information and the edge construction forbid
  certain combinations of nonterminals. We encode the node color information in
 grammar  $\mathit{Pt}^\prime$ and further simplify the language to $D^\prime_1$
 (Section~\ref{subsec:b}).
\item  Finally, the $D^\prime_1$-reachability problem in $G_P$ is isomorphic to the
$D_1$-reachability problem in $G$ (Section~\ref{subsec:c}).
\end{itemize}

\begin{figure}[t]
\centering
\numberwithin{equation}{section} 
{\small
\renewcommand{\theequation}{\arabic{equation}-a} 
\begin{minipage}[b]{.3\linewidth}
\begin{center}
\begin{align} 
\mathit{Pt} \rightarrow &~S~r  \label{rule:bb0}    \\
S \rightarrow &~\mathit{as}~\overline{d}~S~r~d \label{rule:bb1}\\
S \rightarrow &~\overline{d}~\overline{r}~\overline{S}~d~\mathit{sa} \label{rule:bb2}\\
\overline{S} \rightarrow
&~\overline{d}~\overline{r}~\overline{S}~d~\overline{\mathit{as}} \label{rule:bb3}\\
\overline{S} \rightarrow &~\overline{\mathit{sa}}~\overline{d}~S~r~d \label{rule:bb4}\\
S \rightarrow &~S~S~ \label{rule:bb5}\\
\overline{S} \rightarrow &~\overline{S}~\overline{S}~ \label{rule:bb6}\\
S \rightarrow &~s~\mid~\epsilon \label{rule:bb7}\\
\overline{S} \rightarrow &~\overline{s}~\mid~\epsilon \label{rule:bb8}
\end{align}
\end{center}
\subcaption{Rules for language $\mathit{Pt}$.}\label{fig:gpa1}
\end{minipage}%
\renewcommand{\theequation}{\arabic{equation}-b} 
\hfill
\begin{minipage}[b]{.3\linewidth}
\setcounter{equation}{0}
\begin{center}
\begin{align} 
\mathit{Pt}^\prime \rightarrow &~S  \label{rule:cc0}    \\
 \label{rule:cc1}\\
S \rightarrow &~\overline{d}~\overline{r}~\overline{S}~d~\mathit{sa} \label{rule:cc2}\\
 \label{rule:cc3}\\
\overline{S} \rightarrow &~\overline{\mathit{sa}}~\overline{d}~S~r~d \label{rule:cc4}\\
S \rightarrow &~S~S~ \label{rule:cc5}\\
\overline{S} \rightarrow &~\overline{S}~\overline{S}~ \label{rule:cc6}\\
S \rightarrow &~\epsilon \label{rule:cc7}\\
\overline{S} \rightarrow &~\epsilon \label{rule:cc8}
\end{align}
\end{center}
\subcaption{Rules for language $\mathit{Pt}^\prime$.}\label{fig:gpa2}
\end{minipage}%
\hfill
\begin{minipage}[b]{.3\linewidth}
\setcounter{equation}{0}
\renewcommand{\theequation}{\arabic{equation}-c} 
\begin{center}
\begin{align} 
\mathit{Pt}_c^\prime \rightarrow &~\tensor[_\blacksquare]{S}{_\blacksquare}~\\
  \label{rule:dd1}\\
\tensor[_\blacksquare]{S}{_\blacksquare} \rightarrow & 
~{_\blacksquare
  \overline{d}_\square\overline{r}_{\square}}\overline{S}_{\square}d_\square\mathit{sa}_{\blacksquare}\label{rule:dd2}\\
\label{rule:dd3}\\
\tensor[_\square]{\overline{S}}{_\square} \rightarrow & 
~{_\square
  \overline{\mathit{sa}}_\square\overline{d}_{\blacksquare}}S_{\blacksquare}r_\square d_{\square}\label{rule:dd4}\\
\tensor[_\blacksquare]{S}{_\blacksquare} \rightarrow &~_{\blacksquare}S_\blacksquare S_{\blacksquare}\label{rule:dd5}\\
 \label{rule:dd6}\\ 
\tensor[_\blacksquare]{S}{_\blacksquare} \rightarrow &~\epsilon\label{rule:dd7}\\
\tensor[_\square]{\overline{S}}{_\square} \rightarrow &~\epsilon\label{rule:dd8}
\end{align}
\end{center}
\subcaption{Rules for language $\mathit{Pt}_c^\prime$.}\label{fig:gpa3}
\end{minipage}
}
\begin{minipage}[b]{.5\linewidth}
\begin{center}
$\begin{aligned}
\\ 
D_1^\prime \rightarrow &~S \\
\tensor[_\blacksquare]{S}{_\blacksquare} \rightarrow &~\underbracket{_\blacksquare
  \overline{d}_\square\overline{r}_{\square}\overline{\mathit{sa}}_\square\overline{d}_{\blacksquare}}_{[_1}~~\tensor[_\blacksquare]{S}{_\blacksquare}~~\underbracket{_{\blacksquare}r_\square d_{\square}d_\square\mathit{sa}_{\blacksquare}}_{]_1}\\
\tensor[_\blacksquare]{S}{_\blacksquare} \rightarrow &~_{\blacksquare}S_\blacksquare S_{\blacksquare}~\mid~\epsilon\\
\end{aligned}$
\end{center}
\subcaption{Rules for language $D^\prime_1$.}\label{fig:gpa4}
\end{minipage}

\caption{Grammars used in reduction.  \label{fig:gpa}}

\end{figure}

\subsection{$\mathit{Pt}$-Reachability and $\mathit{Pt}^\prime$-Reachability}\label{subsec:a}

\begin{figure}[t]
\centering
\begin{minipage}[b]{0.4\linewidth}
    
\begin{tikzpicture}
        \definecolor{red}{gray}{0.7}
        \node [style={rectangle,inner sep=.8pt, scale = 1.0}]           (a) at (-3,0)  { $\blacksquare$};
        \node [style={rectangle,inner sep=.8pt, scale = 1.0}]           (b) at (-2.375,0)  { $\square$};
        \node [style={rectangle,inner sep=.8pt, scale = 1.0}]           (b1) at (-1.75,0)  { $\square$};
        \node [style={rectangle,inner sep=.8pt, scale = 1.0}]           (c) at (-1.125,0)  { $\square$};
        \node [style={rectangle,inner sep=.8pt, scale = 1.0}]           (d) at (-0.5,0)  { $\blacksquare$};
        \node [style={rectangle,inner sep=.8pt, scale = 1.0}]           (e) at (0.125,0)  { $\square$};
        \node [style={rectangle,inner sep=.8pt, scale = 1.0}]           (e1) at (0.75,0)  { $\square$};
        \node [style={rectangle,inner sep=.8pt, scale = 1.0}]           (f) at (1.375,0)  { $\square$};
        \node [style={rectangle,inner sep=.8pt, scale = 1.0}]           (g) at (2,0)  { $\blacksquare$};

        \node [style={rectangle,inner sep=.8pt, scale = 1.3}]           (h) at (-0.5,0.625)  { $\circ$};
        \node [style={rectangle,inner sep=.8pt, scale = 1.3}]           (i) at (-0.5,1.25)  { $\circ$};
        \node [style={rectangle,inner sep=.8pt, scale = 1.3}]           (j) at (-0.5,1.875)  { $\circ$};
        \node [style={circle,draw, inner sep=2pt, fill=red, scale = 1.0}]           (k) at (-0.5,2.5)  { };
        
        \node [style={rectangle,inner sep=.8pt}, right=1pt of h,scale=.8]           (nodex)   {$x$};
        \node [style={rectangle,inner sep=.8pt}, right=1pt of i,scale=.8]           (nodey)   {$y$};
        \node [style={rectangle,inner sep=.8pt}, right=1pt of j,scale=.8]           (nodez)   {$z$};

        \draw[->] (-2.85, -0.05) --  (-2.525, -0.05);
        \draw[->] (-2.525, 0.05) --  (-2.85, 0.05);
        \draw[->] (-2.225, -0.05) --  (-1.9, -0.05);
        \draw[->] (-1.9, 0.05) --  (-2.225, 0.05);
        \draw[->] (-1.6, -0.05) --  (-1.275, -0.05);
        \draw[->] (-1.275, 0.05) --  (-1.6, 0.05);
        \draw[->] (-0.975, -0.05) --  (-0.65, -0.05);
        \draw[->] (-0.65, 0.05) --  (-0.975, 0.05);
        \draw[->] (-0.35, -0.05) --  (-0.025, -0.05);
        \draw[->] (-0.025, 0.05) --  (-0.35, 0.05);
        \draw[->] (0.275, -0.05) --  (0.6, -0.05);
        \draw[->] (0.6, 0.05) --  (0.275, 0.05);
        \draw[->] (0.9, -0.05) --  (1.225, -0.05);
        \draw[->] (1.225, 0.05) --  (0.9, 0.05);
        \draw[->] (1.525, -0.05) --  (1.85, -0.05);
        \draw[->] (1.85, 0.05) --  (1.525, 0.05);

        \draw[->] (d) -- (h);
        \draw[->] (h) -- (i);
        \draw[->] (i) -- (j);
        \draw[->] (j) -- (k);

        \node (bend1) at (-2, 0.3){};
        \draw[dashed] (a.north) to [out=30, in=180] (bend1.west);
        \node (bend2) at (-1, 0.3){};
        \draw[dashed] (bend1.west) to (bend2.west);
        \node (bend3) at (-0.7, 0.6){};
        \draw[dashed] (bend2.west) to [out=0, in=270] (bend3.west);
        \node (bend4) at (-0.5, 1.45){};
        \draw[dashed] (bend3.west) to [out=90, in=180] (bend4.north);
        \node (bend5) at (-0.3, 0.6){};
        \draw[dashed] (bend4.north) to [out=0, in=90] (bend5.east);
        \node (bend6) at (0, 0.3){};
        \draw[dashed] (bend5.east) to [out=270, in=180] (bend6.east);
        \node (bend7) at (1, 0.3){};
        \draw[dashed] (bend6.east) to (bend7.east);
        \draw[->, dashed] (bend7.east) to [out = 0, in = 150] (g.north);
        
\end{tikzpicture}
\label{figure:dashed_path}
\subcaption{Illustration on a reversing PEG $G_P$ path.}\label{fig:ill}
\end{minipage}
\hspace{0.1\linewidth}
\begin{minipage}[b]{0.4\linewidth}

\begin{tikzpicture}
        \definecolor{red}{gray}{0.7}
        \node [style={rectangle,inner sep=.8pt, scale = 1}]           (a) at (0,0)  { $\blacksquare$};
        \node [style={rectangle,inner sep=.8pt, scale = 1}]           (b) at (1.4,0)  {
          $\circ$};
        \node [style={rectangle,inner sep=.8pt, scale = 1}]           (b1) at (1.4,-1.6)  { $\circ$};
        \node [style={rectangle,inner sep=.8pt, scale = 1}]           (c) at (2.8,-1.6)  { $\circ$};
        \node [style={rectangle,inner sep=.8pt, scale = 1}]           (d) at (4.2,-1.6)  {
          $\circ$};
        \node [style={rectangle,inner sep=.8pt, scale = 1}]           (d1) at (4.2,0)  { $\circ$};
        \node [style={circle,draw, inner sep=2pt, fill=red, scale = 1}]           (e) at (5.6,0)  { };

        \draw[->] (a) -- node[above, midway,scale=.8] {$\mathit{as}$}  (b);
        \draw[->] (b) -- node[left, midway,scale=.8] {$\overline{d}$}  (b1);
        \draw[->] (b1) -- node[above, midway,scale=.8] {$s$}  (c);
        \draw[->] (c) -- node[above, midway,scale=.8] {$r$}  (d); 
        \draw[->] (d) -- node[right, midway,scale=.8] {$d$}  (d1); 
        \draw[->] (d1) -- node[above, midway,scale=.8] {$r$}  (e);

        \node [style={rectangle,inner sep=.8pt}, below=2pt of a,scale=.8]           (node1)   {$1$};
        \node [style={rectangle,inner sep=.8pt}, right=2pt of b,scale=.8]           (node2)   {$2$};
        \node [style={rectangle,inner sep=.8pt}, below=2pt of b1,scale=.8]           (node3)   {$3$};
        \node [style={rectangle,inner sep=.8pt}, below=2pt of c,scale=.8]           (node4)   {$4$};
        \node [style={rectangle,inner sep=.8pt}, below=2pt of d,scale=.8]           (node2)   {$5$};
        \node [style={rectangle,inner sep=.8pt}, left=2pt of d1,scale=.8]           (node2)   {$6$};
        \node [style={rectangle,inner sep=.8pt}, below=2pt of e,scale=.8]           (node7)   {$7$};
        \node [style={rectangle,inner sep=.8pt},above=10pt of a]           (t) {};
        \node [style={rectangle,inner sep=.8pt},below=10pt of c]           (t)  {};
\end{tikzpicture}
\subcaption{PEG $G_P$ path with circle nodes.}\label{figure:cycle}
\label{figure:correspondence_graph}
\end{minipage}
\caption{Irreversibility of circle nodes in PEG $G_P$.}
\end{figure}
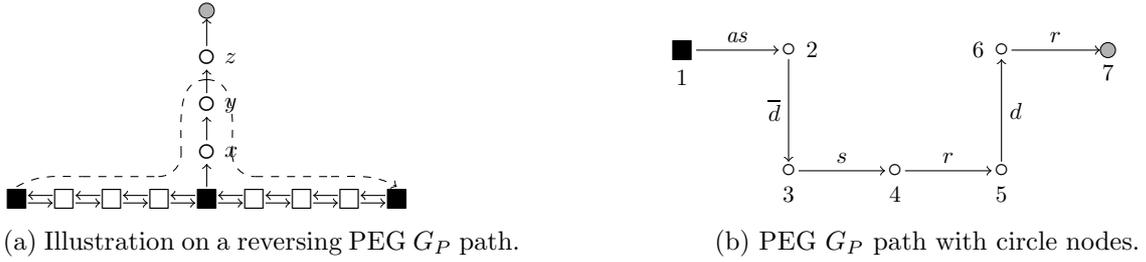

The most
notable difference between $D_1$-reachability and $\mathit{Pt}$-reachability is that 
the $\mathit{Pt}$-reachability problem is \emph{bidirectional}. 
For instance, for any summary edge $u\xrightarrow{S} v$ there exists a reversed summary
$v\xrightarrow{\overline{S}} u$ in $G_P$ based on the $\mathit{Pt}$ grammar in
Figure~\ref{fig:gpa1}. The reversed summaries introduce additional reachability
information since a path now can go back and forth at a node. 

\begin{table}[t]
\begin{center}
\small
\caption{Follow sets for terminals of the $\mathit{Pt}$ language in Figure~\ref{fig:gpa1}.}\label{tab:followset}
\begin{tabular}{ l l | l l }
\hline
Nonterminal & \textsc{Follow} set &Nonterminal & \textsc{Follow} set \\ 
\hline
$\textsc{Follow}(d)$ 
& $\{ sa, \overline{as}, r, d, as, \overline{d}, s, \overline{sa}, \overline{s}  \}$  
& $\textsc{Follow}(\overline{d})$ 
& $\{  as, \overline{d}, s, r, \overline{r}\} $
\\
$\textsc{Follow}(r)$ 
& $\{  d  \}  $
& $\textsc{Follow}(\overline{r})$
& $\{  \overline{d}, d, \overline{sa}, \overline{s}  \}  $
\\
$\textsc{Follow}(as) $
& $\{  \overline{d}  \}  $
& $\textsc{Follow}(\overline{as})$
& $\{  \overline{d},  d, \overline{sa}, \overline{s}   \}  $
\\
$\textsc{Follow}(sa)$
& $\{  r, as, \overline{d}, s  \}$
& $\textsc{Follow}(\overline{sa})$
& $\{  \overline{d}   \} $
\\
$\textsc{Follow}(s)$
& $\{  r, as, \overline{d}, s  \}$
& $\textsc{Follow}(\overline{s})$
& $\{  \overline{d}, d, \overline{sa}, \overline{s}  \} $
\\
\hline 
\end{tabular}

\end{center}
\end{table}

To cope with the bidirectedness, we introduce \emph{reversibility} to paths in $G_P$. Formally, we say a path $p = u,\ldots, x, y, x,\ldots, v$
is a \emph{reversing path} iff there exists at least one node $y \in
p$ such that $x\rightarrow y\rightarrow x$ is a subpath of $p$. The node $y$ is
called a \emph{reversing node} of path $p$. For instance, the path $\blacksquare
\rightarrow \ldots \rightarrow\blacksquare \rightarrow x\rightarrow y\rightarrow
x\rightarrow\blacksquare\rightarrow \ldots \rightarrow \blacksquare$ in Figure~\ref{fig:ill} is a
reversing path with $y$ being the reversing node.

The subpath $x\rightarrow y\rightarrow x$ of a reversing path introduces either a
string ``$t~\overline{t}$'' or a string ``$\overline{t}~t$''. However, most of
those strings are invalid in grammar $\mathit{Pt}$.
Given a $\mathit{CFG} = (\Sigma, N, P, S)$,
we define \textsc{Follow}(t), for terminal $t\in \Sigma$, to be the set of terminals $w$ that can
appear immediately to the right
of terminal $t$ in some sentential form, that is, the set of terminals $w$ such that there
exists a derivation of the form
$A \rightarrow \alpha~t~w~\beta$  for some $\alpha$ and $\beta$.
Note that our definition of \textsc{Follow} set on terminals is similar to the
concept of the Follow
set on nonterminals in standard compiler text~\cite{Aho:86compiler}. 
Table~\ref{tab:followset} gives the \textsc{Follow} sets of all terminals of
grammar $\mathit{Pt}$ in Figure~\ref{fig:gpa4}.

\begin{lemma}\label{lem:reversing}
No circle node can be a reversing node.
\end{lemma}
\begin{proof}
From the $G_P$ construction, we can see that circle nodes represent auxiliary
variables in program $P$. They are used in $G_P$ to connect black and gray nodes.
Figure~\ref{figure:cycle} shows a path with circle nodes.
Based on the \textsc{Follow} sets shown in Table~\ref{tab:followset}, we have only $\overline{d} \in
\textsc{Follow}(d)$. Therefore, in Figure~\ref{figure:cycle}, only nodes $2$ and
$6$ could be the reversing nodes. From Figure~\ref{fig:gpa1}, we can see that the
substring ``$d~\overline{d}$'' can only be generated by rules (\ref{rule:bb5})
and (\ref{rule:bb6}):

\[
S \xRightarrow[\hspace{18pt}]{(\ref{rule:bb5})}    S~S
\xRightarrow[\hspace{18pt}]{(\ref{rule:bb1})} \mathit{as}~\overline{d}~S~r~d~S
\xRightarrow[\hspace{18pt}]{(\ref{rule:bb3})}
\mathit{as}~\overline{d}~S~r~d~\overline{d}~\overline{r}~\overline{S}~d~\mathit{sa};
\]

\[
\overline{S} \xRightarrow[\hspace{18pt}]{(\ref{rule:bb6})}    \overline{S}~\overline{S}
\xRightarrow[\hspace{18pt}]{(\ref{rule:bb4})} \overline{sa}~\overline{d}~S~r~d~\overline{S}
\xRightarrow[\hspace{18pt}]{(\ref{rule:bb3})}
\overline{sa}~\overline{d}~S~r~d~\overline{d}~\overline{r}~\overline{S}~d~\overline{as}.
\]
We notice that there is always a ``$\overline{r}$'' symbol that follows a
``$\overline{d}$'' symbol. Therefore, node $2$ in Figure~\ref{figure:cycle} could not be a
reversing node. 
Without node $2$ being a reversing node, the realized string of path
$1\rightarrow2\rightarrow3\rightarrow4\rightarrow5\rightarrow6\rightarrow5\rightarrow4\rightarrow3\rightarrow2\rightarrow1$
is
``$\mathit{as}~\overline{d}~s~r~d~\overline{d}~\overline{r}~\overline{s}~d~\overline{as}$''. This
realized string cannot be generated by $S$ or $\overline{S}$ discussed
above since $\mathit{as}$ cannot be paired with $\mathit{\overline{as}}$. 
Therefore, node $6$ cannot be a reversing node, either.
\end{proof}

Consider a $G_P$ path with circle nodes shown in Figure~\ref{figure:cycle}. 
It corresponds to the edges generated by the last row in Table~\ref{tab:alledges}.
It contains the variables
introduced in Algorithm~\ref{algo:d1topt} on lines~\ref{algo:pte1}-\ref{algo:pte2}, \myie, node $1$, $2$, $3$, $4$,
$5$, $6$, $7$ represent variables \codeIn{v}, \codeIn{*t}$_i$, \codeIn{t}$_i$,
\codeIn{t}$_{t+1}$, \codeIn{\&t}$_{t+2}$, \codeIn{t}$_{t+2}$ and \codeIn{\&v}$^\prime$ in Algorithm~\ref{algo:d1topt}, respectively.
We say that the black node representing variable \codeIn{v} is a \emph{root node} of the
gray circle node representing variable \codeIn{\&v}$^\prime$ as well as the
white circle nodes representing the auxiliary variables \codeIn{t}$_i$
introduced in Algorithm~\ref{algo:d1topt} on lines~\ref{algo:pte1}-\ref{algo:pte2}.

\begin{lemma}[Isolation]\label{lem:iso}
In $G_P$, the $\mathit{Pt}$- or $S$-path which joins two square nodes cannot pass
through any circle node.
\end{lemma}
\begin{proof}
We prove by contradiction. Assume such a path exists.
Without loss of generality, we assume the path joining  two square
nodes $u$ and $v$ is $u\rightarrow\ldots\rightarrow \circ_x\rightarrow\ldots
\rightarrow v$. Note that node
$\circ_x$ can be either a white circle node or a gray circle node.
Let the root node of $\circ_x$ be $\blacksquare_t$.
Thus, the path is of the form $u\rightarrow,\ldots, \rightarrow \blacksquare_t\rightarrow, \ldots, \circ_x, \ldots,\rightarrow
\blacksquare_t\rightarrow, \ldots, \rightarrow v$. The subpath
$\blacksquare_t\rightarrow, \ldots, \circ_x, \ldots, \rightarrow
\blacksquare_t$ forms a cycle.
As shown in Figure~\ref{figure:cycle}, with the gray node being one end of the
path, there always exits a node $\circ_y$ such that $\blacksquare_t\rightarrow, \ldots, \circ_y, \ldots, \rightarrow
\blacksquare_t$. Thus, $\circ_y$ is a reversing node.
This contradicts the fact that no circle node can be a reversing node (Lemma~\ref{lem:reversing})
\end{proof}

\begin{lemma}\label{lem:ptbg2sbb}
Let $\mathit{Pt}$ and $S$ be two nonterminals in grammar~\ref{fig:gpa1}. In PEG $G_P$, $\mathit{Pt}$-reachability among black and gray nodes is equivalent to
$S$-reachability among black nodes, \myie,
\[
\blacksquare_y\xrightarrow{\mathit{Pt}}\greynode_{\&x^\prime} \Longleftrightarrow
\blacksquare_y\xrightarrow{S}\blacksquare_x.
\]
\end{lemma}
\begin{proof}
Each gray node $\greynode_{\&x^\prime}$ has  one unique root node $\blacksquare_x$.
Due to the $G_P$ construction, there is always a path $\blacksquare_x\xrightarrow{\mathit{as}}\circ_u\xrightarrow{\overline{d}}\circ_v\xrightarrow{s}\circ_w\xrightarrow{r}\circ_y\xrightarrow{d}\circ_z\xrightarrow{r}\greynode_{\&x^\prime}$ between the black node
$\blacksquare_x$ and the gray node $\greynode_{\&x^\prime}$ shown in Figure~\ref{figure:cycle}.
According to the last terminals in rules (\ref{rule:bb1}) and (\ref{rule:bb2}), no $S$-path in PEG
$G_P$ ends at nodes $\circ_u$, $\circ_v$, $\circ_y$, respectively. Moreover,
there is only one $S$-path $\circ_v\xrightarrow{S}\circ_w$ that ends at $\circ_w$,
\myie, there is no node $y$ in $G_P$ such that $y\rightarrow \blacksquare_x
\rightarrow \circ_v\xrightarrow{s}\circ_w$ and $y \xrightarrow{S}\circ_w$.
 Based on rule (\ref{rule:bb1}), we
have $\blacksquare_x\xrightarrow{S}\circ_z\xrightarrow{r}\greynode_{\&x^\prime}$.
\begin{itemize}
\item {The $\Rightarrow$ direction.} Due to the construction, every $\mathit{Pt}$-path
$\blacksquare_y\xrightarrow{\mathit{Pt}}\greynode_{\&x^\prime}$
passes through node $\blacksquare_x$, \myie, there must be a path
  $\blacksquare_y\rightarrow\blacksquare_x\xrightarrow{S}\circ_z\xrightarrow{r}\greynode_{\&x^\prime}$.
Based on rule (\ref{rule:bb0}), we have   $\blacksquare_y\xrightarrow{S}\circ_z\xrightarrow{r}\greynode_{\&x^\prime}$.
The nodes $\circ_u$, $\circ_v$, $\circ_w$,$\circ_y$ between $\blacksquare_x$ and $\circ_z$ are not $S$-reachable from $\blacksquare_y$.
Finally, we have $\blacksquare_y\xrightarrow{S}\blacksquare_x$ based on rule (\ref{rule:bb5}).
\item {The $\Leftarrow$ direction.} For each $S$-path
  $\blacksquare_y\xrightarrow{S}\blacksquare_x$, there is a path
  $\blacksquare_y\xrightarrow{S}\blacksquare_x\xrightarrow{S}\circ_z\xrightarrow{r}\greynode_{\&x^\prime}$. 
Therefore,
  we have a path
  $\blacksquare_y\xrightarrow{\mathit{Pt}}\greynode_{\&x^\prime}$ based on
  rules (\ref{rule:bb0}) and (\ref{rule:bb5}).
\end{itemize}

\end{proof}

\paragraph{Language $\mathit{Pt}^\prime$}
Based on Lemma~\ref{lem:ptbg2sbb}, we are able to compute $S$-reachability with
only black nodes.
Due to Lemma~\ref{lem:iso}, we can discard all circle nodes when computing
$S$-reachability. Therefore, we can discard all rules in Figure~\ref{fig:gpa1}
that contain symbols associated with circle nodes. As a result, we can safely
remove rules (\ref{rule:bb1}), (\ref{rule:bb3}), $S\rightarrow s$ and
$\overline{S}\rightarrow \overline{s}$. Figure~\ref{fig:gpa2} gives the
simplified grammar with a new start symbol $\mathit{Pt}^\prime$. Based on the
discussion, it is immediate that $\mathit{Pt}^\prime$-reachability is equivalent
to $S$-reachability.

\begin{lemma}\label{lem:part1}
In PEG $G_P$, $\mathit{Pt}$-reachability among black and gray nodes is equivalent to
$\mathit{Pt}^\prime$-reachability among black nodes, \myie,
\[
\blacksquare_y\xrightarrow{\mathit{Pt}}\greynode_{\&x^\prime} \Longleftrightarrow
\blacksquare_y\xrightarrow{\mathit{Pt^\prime}}\blacksquare_x.
\]
\end{lemma}

\subsection{$\mathit{Pt}^\prime$-Reachability and $D_1^\prime$-Reachability}\label{subsec:b}

Our basic idea is to ``extract'' a $D_1$ grammar from the
$\mathit{Pt}^\prime$ grammar.  
Every $D_1$ string can be generated by  either rule
$D_1\rightarrow[_1~D_1~]_1$ or rule $D_1\rightarrow D_1~D_1$.
Lemma~\ref{lem:part1} considers  $\mathit{Pt}^\prime$-paths with only black nodes.
However, in $G_P$, the $\mathit{Pt}^\prime$-paths can also
join white nodes. In our proof, we need to make sure that our extracted $D_1$ grammar
only involves black nodes.

\begin{lemma}\label{lem:samecolor}
Based on grammar $\mathit{Pt}^\prime$, the $S$-paths join only
same-color square nodes in $G_P$.
\end{lemma}
\begin{proof}
We prove by contradiction. From Figure~\ref{fig:gpa2}, it is clear that all $S$-paths are
of length $4k$ for some $k\geq 0$. In our $G_P$ construction, there are three
white nodes between a pair of black nodes. We label the three white nodes as
$\blacksquare \rightarrow \square_1\rightarrow \square_2\rightarrow
\square_3\rightarrow \blacksquare$. Assume an $S$-path joins a
black node and a white node. The path could be depicted as one of the followings:
$\blacksquare
\xrightarrow{S} \square_1$,  $\blacksquare
\xrightarrow{S} \square_2$, or $\blacksquare
\xrightarrow{S} \square_3$.
The path lengths are $4k+1$, $4k+2$, and $4k+3$, respectively. Similarly, the
path lengths of $\square_1
\xrightarrow{S} \blacksquare$,  $\square_2
\xrightarrow{S} \blacksquare$, are $\square_3
\xrightarrow{S} \blacksquare$ are $4k+1$, $4k+2$, and $4k+3$ as well.  It
contradicts the fact that $S$-paths are of length $4k$.
\end{proof}
\begin{corollary}
Based on grammar $\mathit{Pt}^\prime$, the $\overline{S}$-paths join only
same-color square nodes in $G_P$.
\end{corollary}

\begin{table}[t]
\caption{All edges} \label{tab:alledges}
\centering
\begin{tabular}{l | c | l}
\hline
Edge type & $G_P$ edge & $P$ statement\\
\hline
\textsc{O1} & $\blacksquare \xrightarrow{\overline{d}} \square
\xrightarrow{\overline{r}} \square \xrightarrow{\mathit{\overline{sa}}}  \square
\xrightarrow{\overline{d}}  \blacksquare$ & \codeIn{t$_i$ = \&x; *y = t$_i$;}\\
\textsc{R1}  &  $\blacksquare \xrightarrow{d} \square \xrightarrow{\mathit{sa}} \square \xrightarrow{r}
\square \xrightarrow{d}  \blacksquare$ & \\
\textsc{O2} & $\blacksquare \xrightarrow{r} \square \xrightarrow{d} \square
\xrightarrow{d}  \square \xrightarrow{\mathit{sa}}  \blacksquare$ & \codeIn{x =
  \&t$_i$; *t$_i$ = y;}\\
\textsc{R2}  &  $\blacksquare \xrightarrow{\mathit{\overline{sa}}} \square \xrightarrow{\overline{d}} \square \xrightarrow{\overline{d}}
\square \xrightarrow{\overline{r}}  \blacksquare$ & \\

\hline 
\end{tabular}

\end{table}

We augment the $\mathit{Pt}^\prime$ grammar in Figure~\ref{fig:gpa2} with node color information in PEG
$G_P$. 
Consider a symbol $\overline{d}$ in Table~\ref{tab:alledges}. 
From Table~\ref{tab:followset}, we can see that symbol $\overline{r} \in
\textsc{Follow}(\overline{d})$. Therefore, a $\overline{d}$-edge can be followed
by an $\overline{r}$-edge. Now, consider an edge $u\xrightarrow{\overline{d}} y$ in
$G_P$. In  Table~\ref{tab:alledges}, there are three types of
$\overline{d}$-edges, \myie, $\square\xrightarrow{\overline{d}}\square$,
$\blacksquare\xrightarrow{\overline{d}}\square$, and $\square\xrightarrow{\overline{d}}\blacksquare$.
It is interesting to note that no $\overline{r}$-edge can follow
$\square\xrightarrow{\overline{d}}\blacksquare$ since all $\overline{r}$-edges start with a
white node $\square$.
To bridge the gap between the color constraints in $G_P$ and the
$\mathit{Pt}^\prime$ grammar, we introduce a \emph{colored form} of grammar $\mathit{Pt}^\prime$.
Specifically, we augment the $\mathit{Pt}^\prime$ grammar in Figure~\ref{fig:gpa2} with node color information in graph
$G_P$. 

\paragraph{Colored Grammar $\tensor[_\blacksquare]{S}{_\blacksquare}$}
Given a $\mathit{CFG} = (\Sigma, N, P, S)$ and a PEG $G_P$, we define a \emph{colored} grammar
$\mathit{CFG}_c = (\Sigma_c, N_c, P_c, S_c)$ where every symbol $\tensor[_l]{t}{_r} \in 
\Sigma_c \cup N_c$ is annotated with two colors $l, r \in \{\square, \blacksquare\}$ iff $l\xrightarrow{t}r \in G_P$ for all $t\in
\Sigma\cup N$. For each production rule $C \rightarrow AB$ in $P$, we construct a
colored rule  $\tensor[_i]{A}{_j} \rightarrow
\tensor[_k]{B}{_l}\tensor[_m]{C}{_n}$ in $P_c$ iff $i=k$, $j=n$, and $l=m$. We
denote it as $\tensor[_i]{A}{_j} \rightarrow
 {_i B_l C_j}$  for
brevity.
The start symbols in $S_c$ could be constructed accordingly. We then give
detailed steps to construct the production rules.
\begin{enumerate}[leftmargin=2cm,rightmargin=1cm,label=Step \arabic*:]
\item From grammar~\ref{fig:gpa2}, we can see that $S$ always begins with a
$\overline{d}$. In Table~\ref{tab:alledges}, there is a unique
  $\blacksquare\xrightarrow{\overline{d}}\square$ in the type \textsc{O1}  edge. Therefore, we have ``$\tensor[_\blacksquare]{S}{_\blacksquare} \rightarrow
\tensor[_\blacksquare]{\overline{d}}{_\square}\ldots$''.
\item In the type \textsc{O1} edge, the unique
  $\blacksquare\xrightarrow{\overline{d}}\square$ is followed by $\square\xrightarrow{\overline{r}}\square$. Therefore, we have ``$\tensor[_\blacksquare]{S}{_\blacksquare} \rightarrow
{_\blacksquare \overline{d}_\square\overline{r}_{\square}}\ldots$''.
\item Based on grammar~\ref{fig:gpa2} and our current construction, there should be an $S$ symbol that begins with a
  $\square$. 
Based on Lemma~\ref{lem:samecolor}, it should be a $\tensor[_\square]{S}{_\square}$.
Our production rule becomes  ``$\tensor[_\blacksquare]{S}{_\blacksquare} \rightarrow
{_\blacksquare \overline{d}_\square\overline{r}_{\square}}S_{\square}\ldots$''. Next,
we look into the last symbols of rule (\ref{rule:cc2}) to complete the construction.
\item The production rule of $\tensor[_\blacksquare]{S}{_\blacksquare}$ should
  end with a $\tensor[_\square]{\mathit{sa}}{_\blacksquare}$. Therefore, we have
  ``$\tensor[_\blacksquare]{S}{_\blacksquare} \rightarrow \ldots
  {_\square\mathit{sa}_{\blacksquare}}$''.
\item The edge $\square \xrightarrow{\mathit{sa}} \blacksquare$ is unique in
  Table~\ref{tab:alledges}. It immediately follows a $\square \xrightarrow{d}
  \square$ in type O2 edge. As a result, we have  ``$\tensor[_\blacksquare]{S}{_\blacksquare} \rightarrow \ldots
  {_\square d_\square\mathit{sa}_{\blacksquare}}$''.
\item Combining steps 3 and 5, we have a complete production rule 
``$\tensor[_\blacksquare]{S}{_\blacksquare} \rightarrow
{_\blacksquare
  \overline{d}_\square\overline{r}_{\square}}S_{\square}d_\square\mathit{sa}_{\blacksquare}$''.
\end{enumerate}

Similarly, we could construct rule ``$\tensor[_\square]{\overline{S}}{_\square} \rightarrow
{_\square
  \overline{\mathit{sa}}_\square\overline{d}_{\blacksquare}}S_{\blacksquare}r_\square
d_{\square}$''. With the beginning $\square$ and the terminal $\overline{sa}$,
we could uniquely locate the ${_\square
  \overline{\mathit{sa}}_\square\overline{d}_{\blacksquare}}$ portion in the
type \textsc{O1} edge of Table~\ref{tab:alledges}. With the beginning
$\blacksquare$ and the terminal $r$, we could also uniquely locate the ${_\blacksquare
  r_\square d_{\square}}$ portion in the
type \textsc{O2} edge in Table~\ref{tab:alledges}. Based on the first and last
symbol in $\tensor[_\blacksquare]{S}{_\blacksquare}$, it is immediate that
$\tensor[_\blacksquare]{S}{_\blacksquare}\rightarrow {_\blacksquare
  S_\blacksquare S_\blacksquare}$.

\begin{lemma}[Non-transitivity]\label{lem:nontran}
$\tensor[_\square]{S}{_\square}$-paths are not transitive.
\end{lemma}
\begin{proof}
We prove by contradiction. Assume that there  is an
$\tensor[_\square]{S}{_\square}$-path in $G_P$ which is generated by two consecutive
$\tensor[_\square]{S}{_\square}$-paths. With the color constraints in Table~\ref{tab:alledges}, we can
see that an $\tensor[_\square]{S}{_\square}$
path can only begin with ${_\square \overline{sa}_\square \overline{d}_{\blacksquare}}$ and
end with ${_\blacksquare r_\square d_{\square}}$. Putting the two
parts together yields a ``$\blacksquare \xrightarrow{r} \square \xrightarrow{d} \square \xrightarrow{\overline{sa}}
\square \xrightarrow{\overline{d}}  \blacksquare$''. It it clear that the
constructed path does not belong to any of the four edge types in
Table~\ref{tab:alledges}. It contradicts the fact that the
$\tensor[_\square]{S}{_\square}$-path is a valid path in $G_P$
\end{proof}

We give the production rules of language $\mathit{Pt}^\prime$ in the colored
form in Figure~\ref{fig:gpa3}.
Note that, based on Lemma~\ref{lem:nontran}, we can eliminate rule~\ref{rule:cc6} in Figure~\ref{fig:gpa2}.
\begin{lemma}\label{lem:part2}
In PEG $G_P$, $\mathit{Pt}^\prime$-reachability among black nodes is equivalent to
$\mathit{Pt}^\prime_c$-reachability, \myie,
\[
\blacksquare_y\xrightarrow{\mathit{Pt}}\blacksquare_x \Longleftrightarrow
\blacksquare_y\xrightarrow{\mathit{Pt^\prime_c}}\blacksquare_x.
\]
\end{lemma}

\paragraph{Language $D^\prime_1$.}  The colored from of language
$\mathit{Pt}^\prime$ in Figure~\ref{fig:gpa3} defines the
$\mathit{Pt}^\prime$-reachability among black nodes in PEG $G_P$. We can
simplify the set of
rules in Figure~\ref{fig:gpa3} by eliminating the nonterminal
$\tensor[_\square]{S}{_\square}$. We obtain the language $D^\prime_1$ in
Figure~\ref{fig:gpa4}. It is immediate that $D^\prime_1$ is equivalent to the
colored form $\mathit{Pt}^\prime_c$ in Figure~\ref{fig:gpa3}. As a result,
$D^\prime_1$-reachability is equivalent to
$\mathit{Pt}^\prime_c$-reachability. 

\subsection{$D_1^\prime$-Reachability and $D_1$-Reachability}\label{subsec:c}
Let $\mathit{Pt}$-reachability refer to the $\mathit{Pt}$-reachability among black
and gray nodes in $G_P$ and $\mathit{Pt}^\prime$-reachability refer to the
$\mathit{Pt}^\prime$-reachability among black nodes in $G_P$.
Based on Lemma~\ref{lem:part1} and Lemma~\ref{lem:part2}, we have:
\[
\mathit{Pt}\text{-reachability} \Longleftrightarrow
\mathit{Pt^\prime}\text{-reachability} \Longleftrightarrow
\mathit{Pt^\prime_c}\text{-reachability} \Longleftrightarrow D^\prime_1\text{-reachability}.
\]

Recall that an $L$-reachability problem instance defined in Definition~\ref{def:cfgreach} contains a digraph $G$ and a
context-free language $\mathit{CFG}$.
The $D^\prime_1$-reachability in PEG $G_P$ is isomorphic to the
$D_1$-reachability in $G$, \myie, there is a bijective mapping between $G_P$ and
$G$ . In particular, each node $v \in G$ has been mapped to a variable in $v \in
\mathit{Var}_b$ in program $P$ based on Algorithm~\ref{algo:d1topt}. The variable has been
constructed as a black square node in $G_P$. Figure~\ref{fig:g2peg} establishes
the bijective mapping between edges. From Figure~\ref{fig:gpa4}, it is clear
that $D^\prime_1$ and $D_1$ are isomorphic. 
Therefore, we have the following
theorem:
\begin{theorem}\label{thm:corr3}
Algorithm~\ref{algo:d1topt} takes as input a digraph $G=(V, E)$ and outputs a
C-style program $P$ with $O(E)$ variables and $O(E)$ statements. All nodes $v\in V$
are represented as variables $\codeInM{v}$ and $\codeInM{v^\prime}$ in $P$.
Node $v$ is $D_1$-reachable from node $u$ in $G$ iff 
the gray node $v^\prime$ is $\mathit{Pt}$-reachable from the black node $u$ in $G_P$.
\end{theorem}

Combining the Lemma~\ref{lem:peg2pt} on the equivalence between
$\mathit{Pt}$-reachability and inclusion-based points-to analysis, we 
prove  Theorem~\ref{thm:corr2}.

\section{Implications of $D_1$-Reachability-Based Reduction}\label{sec:imp}
Our BMM-hardness result of inclusion-based points-to analysis is based on a
reduction  from $D_1$-reachability (Section~\ref{sec1:pt2}).
As mentioned in Section~\ref{sec:intro}, the work by \citet{SridharanF09the} gives a  reduction from transitive closure to
inclusion-based points-to analysis. It is well-known that Boolean matrix
multiplication (BMM) is computationally  equivalent to transitive closure~\cite{FischerM71boolean}.
A natural question that arises is: does 
the $D_1$-reachability-based reduction yield any new insights?
This section discusses two important implications of our reduction.

\begin{itemize}

\item \textbf{Generality.} Based on Table~\ref{tab:ptcons}, points-to analysis on
 C-style programs contains four types of
  constraints: \textsc{Address-of}, \textsc{Assignment}, \textsc{Assign-star}
  and \textsc{Star-assign}. The Sridharan-Fink  reduction only permits
  \textsc{Address-of} and \textsc{Assignment} constraints. 
Based on pointer semantics, 
  real-world C-style programs can have pointers  \emph{without} using any assignment
  statements of the form ``\codeIn{a = b}'', \myie, all assignments are of the
  forms ``\codeIn{a = \&b}'', ``\codeIn{*a = b}'' and ``\codeIn{a = *b}''. Therefore, the
  Sridharan-Fink  reduction does not apply to 
the points-to analysis problem \emph{without} \textsc{Assignment} constraints.
Our reduction based on $D_1$-reachability can be generalized to non-trivial
C-style programs with any types of constraints. It applies to more practical programs (Section~\ref{sec:gen}).

\item \textbf{Expressiveness.} 
Based on $D_1$-reachability, we can establish a more interesting result  that
the demand-driven points-to  analysis is no
easier  than the exhaustive counterpart. As noted in Section~\ref{sec:intro}, the
transitive-closure-based reduction does not imply  such results because the
demand-driven version of graph
reachability can be trivially solved by a linear-time 
depth-first search.
Dyck-reachability is a fundamental framework to formulate many interprocedural program-analysis problems.
Our results demonstrate that the bottleneck of demand-driven interprocedural
analysis is due to matching the well-balanced properties such as procedure
calls/returns and pointer references/dereferences in programs, as opposed to computing the
transitive closure (Section~\ref{sec:hdd}).
\end{itemize}

\subsection{Generality of Reduction}\label{sec:gen}

Consider the four types of  constraints \textsc{Address-of} (\textsc{R}),
\textsc{Assignment} (\textsc{S}),  \textsc{Assign-star} (\textsc{As})
  and \textsc{Star-assign} (\textsc{Sa}) in points-to analysis.
Among the four constraints, the \textsc{Address-of} constraint is essential.
As discussed in Section~\ref{sec:intro}, without \textsc{Address-of}, all points-to sets are empty sets and the
points-to analysis problem becomes trivial. Moreover, if a program contains only
\textsc{Address-of} statements, the points-to sets can be trivially decided in
linear time as there are not any subset constraints.
The other three constraints can be arbitrarily  combined in any practical C-style
programs.
Therefore, to prove Corollary~\ref{cor:uni}, we need to discuss   $ {3 \choose 1} + {3 \choose 2} + {3 \choose 3}
=7$ combinations.  The work by~\citet[\S 3]{SridharanF09the} has already
established the reduction based on  
\textsc{S}. This section focuses on the remaining six cases.

\begin{figure}[t]
\small
\begin{center}
\begin{tabular}{ c|c| c |c}
\hline
Construction Type&Input Graph $G$   & PEG $G_P$ & Program $P$\\
\hline
\multirow{2}{*}[-1cm]{\small \textsc{Edge-with-as}}&
\adjustbox{valign=c}{\begin{tikzpicture}

        \node [style={circle,draw, inner sep=1pt,scale=0.85}]           (a) at (0,0)  { $x$ };
        \node [style={circle,draw, inner sep=1pt,scale=0.8}]           (c) at (1,0)  { $y$ };

   \draw[->] (a) -- node[above, midway,scale=0.8] {$[_1$}  (c);

    \end{tikzpicture}}    &

\adjustbox{valign=c}{\begin{tikzpicture}

        \node [style={rectangle,inner sep=.8pt}]           (a) at (0,0)  { $\blacksquare$};
        \node [style={rectangle,inner sep=.8pt}]           (b) at (0.8,0)  { $\square$};
        \node [style={rectangle,inner sep=.8pt}]           (c) at (0.8,-0.8)  { $\blacksquare$};

        \node [style={rectangle,inner sep=.8pt}, left=2pt of a,scale=.8]           (x)   {
          $x$};
        \node [style={rectangle,inner sep=.8pt}, right=2pt of c,scale=.8]           (y)   {
          $y$};
        \node [style={rectangle,inner sep=.8pt}]           (t) at (-0.8,0) {};
        \node [style={rectangle,inner sep=.8pt},above=10pt of a]           (t) {};
        \node [style={rectangle,inner sep=.8pt},below=10pt of c]           (t) {};

\begin{scope}
   \draw[->] (a) -- node[above, midway,scale=0.7] {$\mathit{as}$}  (b);
   \draw[->] (b) -- node[left, midway,scale=0.7] {$\overline{d}$}  (c);
\end{scope}
    \end{tikzpicture}}
& \shortstack{\codeIn{x = *y;} }
 \\
\cline{2-4}

&\adjustbox{valign=c}{\begin{tikzpicture}

        \node [style={circle,draw, inner sep=1pt,scale=0.85}]           (a) at (0,0)  { $x$ };
        \node [style={circle,draw, inner sep=1pt,scale=0.8}]           (c) at (1,0)  { $y$ };

   \draw[->] (a) -- node[above, midway,scale=0.8] {$]_1$}  (c);

    \end{tikzpicture}}    &

\adjustbox{valign=c}{\begin{tikzpicture}

        \node [style={rectangle,inner sep=.8pt}]           (a) at (0,0)  { $\blacksquare$};
        \node [style={rectangle,inner sep=.8pt}]           (b) at (0.8,0)  { $\square$};
        \node [style={rectangle,inner sep=.8pt}]           (c) at (0.8,0.8)  { $\blacksquare$};

        \node [style={rectangle,inner sep=.8pt}, left=2pt of a,scale=.8]           (x)   {
          $x$};
        \node [style={rectangle,inner sep=.8pt}, right=2pt of c,scale=.8]           (y)   {
          $y$};
        \node [style={rectangle,inner sep=.8pt},]           (t) at (-0.8,0) {};
        \node [style={rectangle,inner sep=.8pt},above=10pt of c]           (t)
              {};
        \node [style={rectangle,inner sep=.8pt},below=10pt of a]           (t)  {};

\begin{scope}
   \draw[->] (a) -- node[above, midway,scale=0.7] {$r$}  (b);
   \draw[->] (b) -- node[right, midway,scale=0.7] {$d$}  (c);
\end{scope}
    \end{tikzpicture}}
 & \shortstack{\codeIn{x = \&y;} } \\
\hline

{\small\textsc{Node-with-as}}&\adjustbox{valign=c}{\begin{tikzpicture}
        \node [style={circle, inner sep=1pt,scale=0.8}]           (c) at (0,0)  { };
        \node [style={circle,draw, inner sep=1pt,scale=0.85}]           (a) at (0.5,0)  { $x$ };
        \node [style={circle, inner sep=1pt,scale=0.8}]           (c) at (1,0)  {  };


    \end{tikzpicture}}    &

\adjustbox{valign=c}{\begin{tikzpicture}
\definecolor{red}{gray}{0.7}
        \node [style={rectangle,inner sep=.8pt}]           (a) at (0,0)  { $\blacksquare$};
        \node [style={rectangle,inner sep=.8pt}]           (b) at (0.8,0)  {
          $\circ$};
        \node [style={rectangle,inner sep=.8pt}]           (b1) at (0.8,-.8)  { $\circ$};

        \node [style={rectangle,inner sep=.8pt}]           (d) at (1.6,-0.8)  {
          $\circ$};
        \node [style={rectangle,inner sep=.8pt}]           (d1) at (1.6,0)  { $\circ$};
        \node [style={circle,draw, inner sep=2pt, fill=red}]           (e) at (2.4,0)  { };

        \node [style={rectangle,inner sep=.8pt}, below=2pt of a,scale=.8]           (x)   {
          $x$};
        \node [style={rectangle,inner sep=.8pt}, below=2pt of e,scale=.8]           (y)   {
          $\&x^\prime$};
        \node [style={rectangle,inner sep=.8pt},above=10pt of a]           (t) {};
        \node [style={rectangle,inner sep=.8pt},below=10pt of b1]           (t)  {};

\begin{scope}
   \draw[->] (a) -- node[above, midway,scale=0.7] {$\mathit{as}$}  (b);
   \draw[->] (b) -- node[left, midway,scale=0.7] {$\overline{d}$}  (b1);

   \draw[->] (b1) -- node[above, midway,scale=0.7] {$r$}  (d); 
   \draw[->] (d) -- node[right, midway,scale=0.7] {$d$}  (d1); 
   \draw[->] (d1) -- node[above, midway,scale=0.7] {$r$}  (e);
\end{scope}
    \end{tikzpicture}}
& \shortstack{\codeIn{x = *t$_i$;} \\ \codeIn{t$_i$ = \&t$_{i+1}$;}
   \\ \codeIn{t$_{i+1}$ = \&x$^\prime$;}} \\

\hline

{\small\textsc{Node-with-s}}&\adjustbox{valign=c}{\begin{tikzpicture}
        \node [style={circle, inner sep=1pt,scale=0.8}]           (c) at (0,0)  { };
        \node [style={circle,draw, inner sep=1pt,scale=0.85}]           (a) at (0.5,0)  { $x$ };
        \node [style={circle, inner sep=1pt,scale=0.8}]           (c) at (1,0)  {  };


    \end{tikzpicture}}    &

\adjustbox{valign=c}{\begin{tikzpicture}
\definecolor{red}{gray}{0.7}
        \node [style={rectangle,inner sep=.8pt}]           (a) at (0,0)  { $\blacksquare$};
        \node [style={rectangle,inner sep=.8pt}]           (b) at (0.8,0)  {
          $\circ$};
        \node [style={circle,draw, inner sep=2pt, fill=red}]           (e) at (1.6,0)  { };

        \node [style={rectangle,inner sep=.8pt}, below=2pt of a,scale=.8]           (x)   {
          $x$};
        \node [style={rectangle,inner sep=.8pt}, below=2pt of e,scale=.8]           (y)   {
          $\&x^\prime$};
        \node [style={rectangle,inner sep=.8pt},above=10pt of a]           (t) {};
        \node [style={rectangle,inner sep=.8pt},below=10pt of a]           (t)  {};

\begin{scope}
   \draw[->] (a) -- node[above, midway,scale=0.7] {$s$}  (b);
   \draw[->] (b) -- node[above, midway,scale=0.7] {$r$}  (e);
\end{scope}
    \end{tikzpicture}}
& \shortstack{\codeIn{x = t$_i$;} \\  \codeIn{t$_{i}$ = \&x$^\prime$;}} \\

\hline 
{\small\textsc{Node-with-path}}&\adjustbox{valign=c}{\begin{tikzpicture}
        \node [style={circle, inner sep=1pt,scale=0.8}]           (c) at (0,0)  { };
        \node [style={circle,draw, inner sep=1pt,scale=0.85}]           (a) at (0.5,0)  { $x$ };
        \node [style={circle, inner sep=1pt,scale=0.8}]           (c) at (1,0)  {  };


    \end{tikzpicture}}    &

\adjustbox{valign=c}{\begin{tikzpicture}
\definecolor{red}{gray}{0.7}
        \node [style={rectangle,inner sep=.8pt}]           (a) at (0,0)  { $\blacksquare$};
        \node [style={circle,draw, inner sep=2pt, fill=red}]           (e) at (1.6,0)  { };

        \node [style={rectangle,inner sep=.8pt}, below=2pt of a,scale=.8]           (x)   {
          $x$};
        \node [style={rectangle,inner sep=.8pt}, below=2pt of e,scale=.8]           (y)   {
          $\&x^\prime$};
        \node [style={rectangle,inner sep=.8pt},above=10pt of a]           (t) {};
        \node [style={rectangle,inner sep=.8pt},below=10pt of a]           (t)  {};

\begin{scope}
   \draw[->] (a) -- node[above, midway,scale=0.7] {$r$}  (e);

\end{scope}
    \end{tikzpicture}}
& \shortstack{\codeIn{x = \&x$^\prime$;}} \\

\hline 

\end{tabular}
\caption{Edge construction in PEG $G_P$ for extended cases.} \label{fig:pegext}
\end{center}
\end{figure}

\begin{figure}[t]
\small
\numberwithin{equation}{section} 

\renewcommand{\theequation}{\arabic{equation}-b'} 
\begin{minipage}[b]{.3\linewidth}
\setcounter{equation}{0}
\begin{center}
\begin{align} 
\mathit{Pt}^\prime \rightarrow &~S  \label{rule:cc01}    \\
S \rightarrow &~\mathit{as}~\overline{d}~S~r~d \label{rule:cc11}\\
\label{rule:cc21}\\
\overline{S} \rightarrow &~\overline{d}~\overline{r}~\overline{S}~d~\overline{\mathit{as}} \label{rule:cc31}\\
\label{rule:cc41}\\
S \rightarrow &~S~S~ \label{rule:cc51}\\
\overline{S} \rightarrow &~\overline{S}~\overline{S}~ \label{rule:cc61}\\
S \rightarrow &~\epsilon \label{rule:cc71}\\
\overline{S} \rightarrow &~\epsilon \label{rule:cc81}
\end{align}
\end{center}
\subcaption{Rules for language $\mathit{Pt}^\prime$.}\label{fig:gpa2e}
\end{minipage}%
\hfill
\renewcommand{\theequation}{\arabic{equation}-c'} 
\begin{minipage}[b]{.3\linewidth}
\setcounter{equation}{0}
\begin{center}
\begin{align} 
\mathit{Pt}_c^\prime \rightarrow &~\tensor[_\blacksquare]{S}{_\blacksquare}~  \label{rule:cc0e}    \\
\tensor[_\blacksquare]{S}{_\blacksquare} \rightarrow & 
~{_\blacksquare
  \mathit{as}_\square\overline{d}_{\blacksquare}}S_{\blacksquare}r_{\square}d_{\blacksquare} \label{rule:cc1e}\\
~\label{rule:cc2e}\\
~ \label{rule:cc3e}\\
~\label{rule:cc4e}\\
\tensor[_\blacksquare]{S}{_\blacksquare} \rightarrow &~_{\blacksquare}S_\blacksquare S_{\blacksquare}\label{rule:cc5e}\\
~ \label{rule:cc6e}\\
\tensor[_\blacksquare]{S}{_\blacksquare} \rightarrow &~\epsilon \label{rule:cc7e}\\
~ \label{rule:cc8e}
\end{align}
\end{center}
\vspace{.2cm}
\subcaption{Rules for language $\mathit{Pt}_c^\prime$.}\label{fig:gpa3e}
\end{minipage}%
\hfill
\begin{minipage}[b]{.3\linewidth}
\begin{center}
$\begin{aligned}
\\ 
D_1^\prime \rightarrow &~S \\
\tensor[_\blacksquare]{S}{_\blacksquare} \rightarrow &~\underbracket{_\blacksquare
  \mathit{as}_\square\overline{d}_{\blacksquare}}_{[_1}~~\tensor[_\blacksquare]{S}{_\blacksquare}~~\underbracket{_{\blacksquare}r_\square d_{\blacksquare}}_{]_1}\\
\tensor[_\blacksquare]{S}{_\blacksquare} \rightarrow &~_{\blacksquare}S_\blacksquare S_{\blacksquare}~\mid~\epsilon\\
\end{aligned}$
\end{center}
\vspace{2cm}
\subcaption{Rules for language $D^\prime_1$.}\label{fig:gpa4e}
\end{minipage}

\caption{Grammars used for extended cases. The $\mathit{Pt}^\prime$ grammar in
  Figure~\ref{fig:gpa2e} is obtained from the $\mathit{Pt}$ grammar in Figure~\ref{fig:gpa1} by
  removing two rules related to $\mathit{sa}$ and $\overline{sa}$.  \label{fig:gpaee}}

\end{figure}

\begin{itemize}
\item {\bf Case 1 with constraints \textsc{Sa}, \textsc{As}, and \textsc{S}.}
  Section~\ref{sec1:pt2} has established the construction. It gives a reduction based on \emph{all} three constraints, which
   is the main contribution of this paper. The other five
   cases are indeed extended from this case. We briefly
summarize the reduction and the correctness to facilitate  the discussions on
  other cases. 
\begin{itemize}
\item \emph{Reduction.}  Figure~\ref{fig:g2peg} gives the construction. For each edge in $G$,
   we use \textsc{Edge-with-sa} to construct edges in $G_P$ with square
   ($\square$ or $\blacksquare$) nodes. For each node in $G$, we construct paths based on
   \textsc{Node-with-as-s} using circle (\tikz\draw[black,fill=white] (0,0)
   circle (.5ex); or \tikz\draw[black,fill=gray,opacity=0.8] (0,0) circle
   (.5ex);) nodes.
\item \emph{Correctness.} Based on Section~\ref{sec1:pt2}, the key steps to prove the
  correctness (Theorem~\ref{thm:corr3}) include establishing the isolation lemma (Lemma~\ref{lem:iso}) and the non-transitivity
  lemma (Lemma~\ref{lem:nontran}). From Figure~\ref{figure:reduction_overview},
  we can see that the isolation lemma
  guarantees that $\mathit{Pt}$-reachability is equivalent to
  $\mathit{Pt}^\prime$-reachability among only black square $\blacksquare$ nodes.
The non-transitivity lemma ensures that the mapping between 
$\mathit{Pt}^\prime$-paths and $D_1^\prime$-paths is bijective.
Finally, due to the grammar construction, $D_1^\prime$-reachability in $G_P$ is
always isomorphic to $D_1$-reachability in $G$.
\end{itemize}
\item {\bf Case 2 with constraints \textsc{Sa} and \textsc{S}.}  Case 1 includes
  all three  constraints. In Case 2, we need to construct a program without any \textsc{As}
  constraint. Based on Figure~\ref{fig:g2peg}, we can see that the \textsc{As} constraint only appears at the
  paths for node construction (\textsc{Node-with-as-s}). Therefore, in Case 2,
  we only need to remove the $\mathit{as}$-related nodes/edges when constructing nodes in $G_P$.
\begin{itemize}
\item  \emph{Reduction.} For each  edge in $G$, we adopt 
  \textsc{Edge-with-sa} shown in Figure~\ref{fig:g2peg}. For each node in $G$, we
  construct a path with \textsc{Node-with-s} shown in Figure~\ref{fig:pegext}. Comparing with
  \textsc{Node-with-as-a} in Case 1, \textsc{Node-with-s} removes the edges representing
  \textsc{As} and the corresponding reference/dereference.
\item  \emph{Correctness.}  According to  Case 1, it suffices to show that
  Lemmas~\ref{lem:iso} and~\ref{lem:nontran} hold for Case 2.
The nodes used in \textsc{Node-with-s} construction are circle (\tikz\draw[black,fill=white] (0,0)
   circle (.5ex); or \tikz\draw[black,fill=gray,opacity=0.8] (0,0) circle
   (.5ex);) nodes.
 Lemma~\ref{lem:reversing} holds for all cycle nodes. The \textsc{Node-with-s} construction only removes
  edges in \textsc{Node-with-as-a}. Therefore, it holds
  for the \textsc{Node-with-s} construction as well. 
  This concludes that Lemma~\ref{lem:iso} holds. Note that the \textsc{Edge-with-sa} construction
  is identical to the edge construction in Case 1 (Lemma~\ref{lem:nontran}). 
\end{itemize}
\item {\bf Case 3 with constraints \textsc{Sa} and \textsc{As}.} Similar to Case
  2, we only need to avoid using the \textsc{S} constraint when constructing $G_P$.
\begin{itemize}
\item  \emph{Reduction.} For each edge in $G$, we use
  \textsc{Edge-with-sa} in Figure~\ref{fig:g2peg} as the previous two cases. For each node in
  $G$, we construct a path with \textsc{Node-with-as} in Figure~\ref{fig:pegext}. Comparing with
  \textsc{Node-with-as-s} in Case 1, we remove the edges representing
  \textsc{S} in Case 3.
\item  \emph{Correctness.} Both Cases 2 and 3 avoid using one constraint
  compared with the node construction in Case 1. Following a
  similar argument as in Case 2, Lemmas~\ref{lem:iso} and~\ref{lem:nontran} hold.
\end{itemize}
\item {\bf Case 4 with constraints \textsc{As} and \textsc{S}.}  In Case 1, we
use $\mathit{sa}$-edges in $G_P$ to model ``$[_1$''- and ``$]_1$''-edges in
$G$. However, Case
4 does not contain any \textsc{Sa} constraints. We need to use
$\mathit{as}$-edges to encode ``$[_1$''- and ``$]_1$''-edges in $G$. The principal idea
is based on Case 1.
Indeed, due to the $\mathit{Pt}$ grammar, the reduction based on
$\mathit{as}$-edges is simpler than Case 1 which is based on $\mathit{sa}$-edges.
Let us revisit 
  the $\mathit{Pt}$ grammar in Figure~\ref{fig:gpa1}. The reduction in
  Section~\ref{subsec:a} uses rules~\ref{rule:bb2} and~\ref{rule:bb4} to encode the
  $D_1$ rule $S \rightarrow [_1~~S~~]_1$ shown in Figure~\ref{fig:gpa4}. Without the
  $\mathit{sa}$ terminal, we can still use rule~\ref{rule:bb1}. It is
  interesting to note that,
  unlike rule~\ref{rule:bb2}, the $S$ nonterminal in rule~\ref{rule:bb1} only
  depends on $S$ itself. Rule~\ref{rule:bb1} is immediately isomorphic to  $D_1$ rule $S
  \rightarrow [_1~~S~~]_1$. Consequently, the reduction based on
  $\mathit{as}$-edges is significantly simpler because rule~\ref{rule:bb1} does
  not contain the $\overline{S}$ nonterminal.
Figure~\ref{fig:gpaee} gives all grammars used in the new reduction for Case 4.
Figure~\ref{fig:gpa4e} gives the grammar that is isomorphic to the $D_1$ grammar. 
\begin{itemize}
\item \emph{Reduction.} For each labeled edge in $G$, we use
  \textsc{Edge-with-as} in Figure~\ref{fig:pegext} to construct edges in $G_P$. For each node in $G$, we
   construct a path based on \textsc{Node-with-s} in Figure~\ref{fig:pegext}. 
\item \emph{Correctness.} It is clear
from Figure~\ref{fig:gpa2e} that rules~\ref{rule:cc31} and ~\ref{rule:cc61} are
unreachable from the start symbol, which can be safely
eliminated.    Due to the \textsc{Edge-with-as} and grammar $D_1^\prime$ in
Figure~\ref{fig:gpa4e}, there is no ``$\tensor[_\square]{S}{_\square}$''
symbol, \myie, no white $\square$ nodes $S$-reachable in the new reduction.
Therefore, the
non-transitive Lemma~\ref{lem:nontran} holds immediately. The
\textsc{Node-with-s} construction has been used in Case 2 and  Lemma~\ref{lem:iso} holds.
\end{itemize}
\item {\bf Case 5 with constraint \textsc{Sa}} This case avoids using
  constraint \textsc{As} in  Case 3.
\begin{itemize}
\item  \emph{Reduction.} For each edge in $G$, we use
  \textsc{Edge-with-sa} in Figure~\ref{fig:g2peg} as Cases 1-3. For each node in
  $G$, we construct  a path with \textsc{Node-with-path} in Figure~\ref{fig:pegext}. Comparing with
  \textsc{Node-with-s} in Case 2, we remove the edge representing
  \textsc{S} in Case 5.
\item  \emph{Correctness.} The \textsc{Edge-with-sa} construction
  is identical to the edge construction in Case 1
  (\myie, Lemma~\ref{lem:nontran} holds). There is only one edge in the
  \textsc{Node-with-path} construction. Based on $\textsc{Follow}(r)$ and
  $\textsc{Follow}(\overline{r})$ in Table~\ref{tab:followset}, Lemma~\ref{lem:iso}
 holds.
\end{itemize}
\item {\bf Case 6 with constraint \textsc{As}} This case avoids using
  constraint \textsc{Sa} in  Case 3.
\begin{itemize}
\item  \emph{Reduction.} For each edge in $G$, we use
  \textsc{Edge-with-as} in Figure~\ref{fig:pegext} as Case 4. For each node in
  $G$, we construct  a path with \textsc{Node-with-path} in
  Figure~\ref{fig:pegext}, which is identical to  Case 5.
\item  \emph{Correctness.}  The \textsc{Node-with-path}
  construction is identical to the node construction in Case 5 (\myie Lemma~\ref{lem:iso}
 holds). The \textsc{Edge-with-as} construction
  is identical to the edge construction in Case 4
  (\myie, Lemma~\ref{lem:nontran} holds).
\end{itemize}
\end{itemize}

Putting everything together, we prove Corollary~\ref{cor:uni} in Section~\ref{sec:intro}.

\subsection{Hardness of Demand-Driven Analysis}\label{sec:hdd}

Section~\ref{sec1:pt2} gives a reduction from $D_1$-reachability to points-to analysis. 
Corollary~\ref{cor:dd} further states that $D_1$-reachability can be reduced to points-to
analysis under arbitrary  combinations of statement types.
Therefore, to
establish the BMM-hardness of demand-driven points-to analysis, it suffices to
establish the BMM-hardness of  single-source-single-target
$D_1$-reachability ($s$-$t$ $D_1$-reachability).
Due to the subcubic fine-grained
equivalence of BMM and Triangle Detection~\cite[Thm 1.3]{WilliamsW18subcubic}, the
BMM conjecture is equivalent to: 
\begin{conjecture}
Any combinatorial algorithm for Triangle Detection
in 
graphs with $n$ nodes requires $n^{3-o(1)}$ time in the Word-RAM model of computation with $O(\log n)$ bit words.
\end{conjecture}We pick the problem of Triangle Detection because the reduction is more
intuitive. Similar to BMM, Triangle Detection has widely been used in fine-grained
complexity proofs~\cite{WilliamsW18subcubic,AbboudW14popular,Patrascu10towards}.
This section gives a reduction from Triangle Detection to $s$-$t$ $D_1$-reachability.

\begin{problem}[framed]{\large Reduction: From Triangle Detection to $s$-$t$ $D_1$-Reachability}
  Input: & An (un)directed graph $G$ with $n$ nodes and $m$ edges; \\
  Output: & An edge-labeled digraph $G^\prime= (V^\prime, E^\prime)$, where $|V^\prime| = 4n+6m+2$ and $|E^\prime|
  = 2n+12m$. 
\end{problem}
\paragraph{Intuition.} 

We introduce two unique nodes $s$ and $t$ in the output graph $G^\prime$. The
input graph $G$ contains a triangle iff node $t$ is $D_1$-reachable from $s$.
We arrange all nodes based on a particular ordering and split each node  $u \in G$ into four copies
$u_0$, $u_1$,
$u_2$ and $u_3$ that span four layers in $G^\prime$. The four layers  with three
edges in $G^\prime$ accompany  a triangle  in $G$.
The graph structure is informally known as a
tripartite
graph~\cite{WilliamsW18subcubic,AbboudW14popular,Patrascu10towards}. In
particular, the graph contains two parts:
%
\begin{itemize}
\item \emph{Triangle finding part:} For each $u\rightarrow v$ in $G$, we connect $u$ and $v$ in two adjacent layers
in $G^\prime$, \myie,
 $u_j \xrightarrow{D_1} v_{j+1}$ and $v_{j} \xrightarrow{D_1} u_{j+1}$ for all
$j\in [0,3)$. Therefore, we have a bijective map between 
$u\rightarrow v\rightarrow w\rightarrow u$ in $G$ and
 $u_0\xrightarrow{D_1} v_1\xrightarrow{D_1}
w_2\xrightarrow{D_1} u_3$ in $G^\prime$.
\item \emph{Existential testing part:} To test the existence  of a triangle in
  $G$, we connect all nodes in layer 0  via ``$[_1$-edges'' and all nodes in layer 3
  via ``$]_1$-edges''. Since all nodes in $G$ are arranged based on a random
  ordering, let $x$ be the first node that appears in the particular  ordering.  We
  construct $s\xrightarrow{[_1}x_0$ and $x_3\xrightarrow{]_1}t$. Therefore, if graph $G$ has a triangle that contains
  node $u$, there exists a corresponding $D_1$-path $s\xrightarrow{[_1}\ldots\xrightarrow{[_1}
  u_0\xrightarrow{D_1} u_3\xrightarrow{]_1} \ldots \xrightarrow{]_1} t$ in $G^\prime$, and vice versa.
\end{itemize}

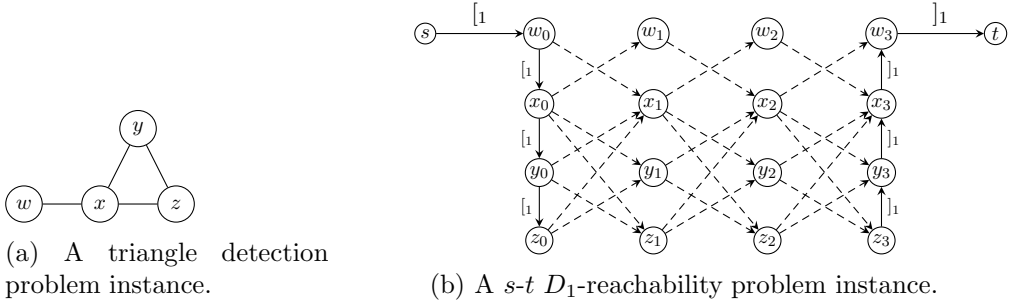
\begin{figure}[!t]
  \centering
  \begin{subfigure}[b]{0.25\textwidth}
\small
\begin{tikzpicture}[every path/.style={>=latex}]


  \node [style={draw,circle, inner sep=1pt, minimum size=0.6cm, scale=0.8}]
  (a) at (0,0)  {$w$ };
  \node [style={draw,circle, inner sep=1pt, minimum size=0.6cm, scale=0.8}]
  (a1) at (1,0)  {$x$ };
  \node [style={draw,circle, inner sep=1pt, minimum size=0.6cm, scale=0.8}]
  (a2) at (1.5,1)  {$y$ };
  \node [style={draw,circle, inner sep=1pt, minimum size=0.6cm, scale=0.8}]
  (a3) at (2,0)  {$z$ };

  \draw[-] (a) edge  (a1);
  \draw[-] (a1) edge  (a2);
  \draw[-] (a2) edge  (a3);
  \draw[-] (a1) edge  (a3);

\end{tikzpicture}

    \caption{A triangle detection problem instance.\label{fig:demand1}}
  \end{subfigure}%
  \begin{subfigure}[b]{0.55\textwidth}
\hspace{1cm}
\begin{tikzpicture}[>=stealth]
\newcommand{\myw}{15pt}
\newcommand{\myis}{{1pt}}
\newcommand{\myscale}{{0.7}}

  \node [style={draw,circle, inner sep=\myis,  scale=\myscale}]
  (a) at (0,0)  {$w_0$ };
  \node [style={draw,circle, inner sep=\myis,  scale=\myscale}]
  (a1) at (1.5,0)  {$w_1$ };
  \node [style={draw,circle, inner sep=\myis,  scale=\myscale}]
  (a2) at (3,0)  {$w_2$ };
  \node [style={draw,circle, inner sep=\myis,  scale=\myscale}]
  (a3) at (4.5,0)  {$w_3$ };
 
  \node [style={draw,circle, inner sep=\myis,  scale=\myscale},below = \myw of a]
  (b)   {$x_0$ };
  \node [style={draw,circle, inner sep=\myis,  scale=\myscale}, below = \myw of a1]
  (b1)   {$x_1$ };
  \node [style={draw,circle, inner sep=\myis,  scale=\myscale},below = \myw of a2]
  (b2)  {$x_2$ };
  \node [style={draw,circle, inner sep=\myis,  scale=\myscale},below = \myw of a3]
  (b3)   {$x_3$ };

  \node [style={draw,circle, inner sep=\myis,  scale=\myscale},below = \myw of b]
  (c)   {$y_0$ };
  \node [style={draw,circle, inner sep=\myis,  scale=\myscale},below = \myw of b1]
  (c1)   {$y_1$ };
  \node [style={draw,circle, inner sep=\myis,  scale=\myscale},below = \myw of b2]
  (c2)   {$y_2$ };
  \node [style={draw,circle, inner sep=\myis,  scale=\myscale},below = \myw of b3]
  (c3)   {$y_3$ };

  \node [style={draw,circle, inner sep=\myis,  scale=\myscale},below = \myw of c]
  (d)  {$z_0$ };
  \node [style={draw,circle, inner sep=\myis,  scale=\myscale},below = \myw of c1]
  (d1)   {$z_1$ };
  \node [style={draw,circle, inner sep=\myis, scale=\myscale},below = \myw of c2]
  (d2)   {$z_2$ };
  \node [style={draw,circle, inner sep=\myis, scale=\myscale},below = \myw of c3]
  (d3)   {$z_3$ };

  \node [style={draw,circle, inner sep=1.5pt,  scale=\myscale}]
  (s) at (-1.5,0)  {$s$ };
  \node [style={draw,circle, inner sep=1.5pt,  scale=\myscale}]
  (t) at (6,0)  {$t$ };

  \draw[->] (s) edge node[above, scale=0.8] {$[_1$}(a);

  \draw[->] (a3) edge node[above, scale=0.8] {$]_1$}(t);

  \draw[->] (a) edge node[left, scale=0.6] {$[_1$}(b);

  \draw[->] (b) edge node[left, scale=0.6] {$[_1$}(c);
  \draw[->] (c) edge node[left, scale=0.6] {$[_1$}(d);

  \draw[->] (b3) edge node[right, scale=0.6] {$]_1$}(a3);

  \draw[->] (c3) edge node[right, scale=0.6] {$]_1$}(b3);
  \draw[->] (d3) edge node[right, scale=0.6] {$]_1$}(c3);

  \draw[densely dashed,->] (a) edge (b1);
  \draw[densely dashed,->] (b) edge (a1);
  \draw[densely dashed,->] (a1) edge (b2);
  \draw[densely dashed,->] (b1) edge (a2);
  \draw[densely dashed,->] (a2) edge (b3);
  \draw[densely dashed,->] (b2) edge (a3);

  \draw[densely dashed,->] (b) edge (c1);
  \draw[densely dashed,->] (c) edge (b1);
  \draw[densely dashed,->] (b1) edge (c2);
  \draw[densely dashed,->] (c1) edge (b2);
  \draw[densely dashed,->] (b2) edge (c3);
  \draw[densely dashed,->] (c2) edge (b3);

  \draw[densely dashed,->] (c) edge (d1);
  \draw[densely dashed,->] (d) edge (c1);
  \draw[densely dashed,->] (c1) edge (d2);
  \draw[densely dashed,->] (d1) edge (c2);
  \draw[densely dashed,->] (c2) edge (d3);
  \draw[densely dashed,->] (d2) edge (c3);

  \draw[densely dashed,->] (b) edge (d1);
  \draw[densely dashed,->] (d) edge (b1);
  \draw[densely dashed,->] (b1) edge (d2);
  \draw[densely dashed,->] (d1) edge (b2);
  \draw[densely dashed,->] (b2) edge (d3);
  \draw[densely dashed,->] (d2) edge (b3);

\end{tikzpicture}

    \caption{A $s$-$t$ $D_1$-reachability problem instance.\label{fig:demand2}}
  \end{subfigure}

\caption{Reduction from Triangle Detection to $s$-$t$ $D_1$-reachability. For
  brevity, we use the dashed edges $u \dashrightarrow v$ to denote
  $u\xrightarrow{[_1}t_i\xrightarrow{]_1}v$ for some auxiliary nodes $t_i$.}\label{fig:demand}
\end{figure}

\begin{algorithm}[!t]
\footnotesize
\SetKwData{Null}{Null}
\SetKwData{oldIn}{oldIn}
\SetKwData{Up}{up}
\SetKwData{current}{current$_u$}
\SetKwData{previous}{previous$_u$}
\SetKwData{lastc}{last$_u$}
\SetKwData{lastp}{last$_p$}
\SetKwData{outaf}{\textsc{Out}$_{\alpha}$}
\SetKwData{outd}{out$_d$}
\SetKwData{outdb}{out$_{\bar{d}}$}
\SetKwData{rowu}{Row$_u$}

\SetKwFunction{Add}{Add}
\SetKwFunction{Init}{Init}
\SetKwFunction{PF}{ParititionFunc}
\SetKwInOut{Input}{Input}
\SetKwInOut{Output}{Output}
\DontPrintSemicolon

\Input{An undirected graph $G=(V, E)$;  }
\Output{An edge-labeled graph $G^\prime$.}

\BlankLine

Introduce four nodes $v_0$, $v_1$, $v_2$, and $v_3$ to $G^\prime$ for all $v \in
V$ \nllabel{algo:ext1}\\

Introduce two unique source and sink nodes $s$ and $t$ to $G^\prime$ \nllabel{algo:ext2}\\
$i \gets 0$ \\

$u \gets \textsc{Select-Node} (V)$ and $ V \leftarrow V
\setminus \{u\}$ \nllabel{algo:ext3}\\
Insert edges $s\xrightarrow{[_1}u_0$ and $u_3\xrightarrow{]_1}t$ to $G^\prime$ \nllabel{algo:line:addeb}\\

\While{$V \neq \emptyset$}{
$\mathit{last}_0 \leftarrow u_0$ and $\mathit{last}_3 \leftarrow u_3$ \\
$u \gets \textsc{Select-Node}(V)$ and $ V \leftarrow V
\setminus \{u\}$ \\
Insert edges $\mathit{last}_0\xrightarrow{[_1}u_0$ and
  $u_3\xrightarrow{]_1}\mathit{last}_3$ to $G^\prime$  \nllabel{algo:line:addee}\\
}

\ForEach{edge $ (u, v)\in E$\nllabel{algo:ext22}}{

  \For{$j \gets 0$ to 2} {
    Introduce two auxiliary nodes $t_i, t^\prime_i$ to $G^\prime$\\
    Insert edges $u_j\xrightarrow{[_1}t_i$ and $t_i\xrightarrow{]_1}v_{j+1}$ to
    $G^\prime$ \nllabel{algo:line:adde}\\
    Insert edges $v_j\xrightarrow{[_1}t^\prime_i$ and
      $t^\prime_i\xrightarrow{]_1}u_{j+1}$ to $G^\prime$ \tcp{{\scriptsize Omit this line if
      the input graph is directed.}} \nllabel{algo:line:adde1}
 $i \gets i+1$\nllabel{algo:ext33}\\
  }

}

\caption{Reduction from Triangle detection to $s$-$t$ $D_1$-reachability.}\label{algo:std1}

\end{algorithm}

\paragraph{Reduction.} Algorithm~\ref{algo:std1} gives the reduction which takes
as input  a
graph $G$ with $n$ nodes and $m$ edges. For each node in $G$, we introduce four nodes (lines~\ref{algo:ext1}-\ref{algo:ext2}) and two edge in
the tripartite
graph $G^\prime$ (lines~\ref{algo:ext3}-\ref{algo:line:addee}). For each edge in
$G$, we introduce six nodes and 12 edges (line~\ref{algo:ext22}-\ref{algo:ext33}).
Therefore, algorithm~\ref{algo:std1} outputs a digraph $G^\prime$ with $4n+6m+2$ nodes
and $2n+12m$ edges.

\paragraph{Correctness.}
It is clear that Algorithm~\ref{algo:std1} is a linear-time
reduction in terms of the input graph size. Specifically, the output tripartite
graph contains $O(m+n)$ nodes and $O(m+n)$ edges. Therefore, it is a subcubic reduction.
To show $s$-$t$ $D_1$-reachability is BMM-hard, it suffices to prove the following
lemma on reduction correctness.

\begin{lemma}\label{lem:bmmstd1}
Algorithm~\ref{algo:std1}  is a linear-time reduction which takes as
input an undirected graph $G$ and outputs a digraph $G^\prime$ with two unique
nodes $s$ and $t$. Graph $G$ has a triangle iff node $t$ is $D_1$-reachable from
$s$ in $G^\prime$.
\end{lemma}
\begin{proof}
A triangle in graph $G$ corresponds to a $D_1$-path from $s$ to $t$ in $G^\prime$ and vice versa.
\begin{itemize}
\item \emph{The ``$\Rightarrow$'' direction.} Without out loss of generality, we
  assume a triangle $x\rightarrow y \rightarrow z$ in $G$. 
Algorithm~\ref{algo:std1} (line~\ref{algo:line:adde}) introduces a  corresponding
path $x_0\xrightarrow{D_1}y_1\xrightarrow{D_1}z_2\xrightarrow{D_1}x_3$ in
$G^\prime$. Therefore, $x_3$ is $D_1$-reachable from $x_0$ in $G^\prime$. From lines
\ref{algo:line:addeb}-\ref{algo:line:addee}, we can see that
Algorithm~\ref{algo:std1} introduces a path $p_1$ from $s$ to \emph{any} node $u_0$ using
``$[_1$'' labels and a path  $p_2$ from the corresponding node $u_3$ to $t$
  using ``$]_1$'' labels. The number of open brackets in $p_1$ equals to the number of
close brackets in $p_2$. Since $x_3$ is $D_1$-reachable from $x_0$ and the brackets in
$p_1$ and $p_2$ are properly matched, node $t$ is $D_1$-reachable from $s$ in $G^\prime$.

\item \emph{The ``$\Leftarrow$'' direction.} 
Similar to the construction in Section~\ref{subsec:bmm}, our constructed graph $G^\prime$ is a
  $4$-layered graph, \myie, it contains four node sets $V_0$, $V_1$,
$V_2$, and $V_3$. 
Algorithm~\ref{algo:std1} (lines \ref{algo:line:adde} and
\ref{algo:line:adde1}) always introduces paths with properly-matched brackets
from nodes in $V_0$ to nodes in $V_3$. 
Suppose there exists a $D_1$-path from  $s$ to $t$ in $G^\prime$. The path
contains three sub-paths: (1) sub-path $p_1$ from $s$ to some node $u_0$; (2)
sub-path $p_2$ from $u_0$ to $v_3$;
and (3) sub-path $p_3$ from $v_3$ to $t$. Moreover, $p_1$ contains
unmatched open brackets and $p_3$ contains unmatched close brackets. 
Due to Algorithm~\ref{algo:std1} (lines
\ref{algo:line:addeb}-\ref{algo:line:addee}), the brackets in $p_1$ and $p_3$
can match iff $u_0 = v_0$ or $v_3 = u_3$. As a result, the $D_1$-path is of the
form $s\xrightarrow{[_1}\ldots\xrightarrow{[_1}u_0\ldots u_3\xrightarrow{]_1}\ldots\xrightarrow{]_1}t$.
Because  graph $G^\prime$ does not contain any cycle, the path joining $u_0$ and
$u_3$ must be of the form $u_0\xrightarrow{D_1}x_1\xrightarrow{D_1}y_2\xrightarrow{D_1}u_3$.
It corresponds to a triangle $u\rightarrow x\rightarrow y\rightarrow u$ in $G$.
\end{itemize}\end{proof}

Section~\ref{sec1:pt2} gives a reduction from \emph{all-pairs}
$D_1$-reachability to \emph{all-pairs} $\mathit{Pt}$-reachability (Theorem~\ref{thm:corr3}).
The \emph{all-pairs} $\mathit{Pt}$-reachability is equivalent to
\emph{exhaustive} points-to analysis (Lemma~\ref{lem:peg2pt}). This section establishes
the BMM-hardness of $s$-$t$ $D_1$-reachability
(Lemma~\ref{lem:bmmstd1}). Putting everything together, we prove that
demand-driven points-to analysis is BMM-hard (Corollary~\ref{cor:dd}). $D_1$-reachability essentially captures the balanced-parenthesis property of
pointer references and dereferences. It is worth noting that Corollary~\ref{cor:dd}
only holds for non-trivial programs with pointer dereferences. For programs
without dereferences, the demand-driven points-to analysis can be
solved via a linear-time depth-first search based on the Sridharan-Fink reduction~\cite{SridharanF09the}.

\paragraph{Demand-Driven Interprocedural Program Analysis.}
Lemma~\ref{lem:bmmstd1} establishes the BMM-hardness of $s$-$t$
$D_1$-reachability. 
It is immediate that Dyck-reachability with $k$ kinds of parenthesis
($s$-$t$ $D_k$-reachability) is BMM-hard. $D_k$-reachability is a fundamental
framework to describe interprocedural program analysis
problems~\cite{Reps97program,RepsHS95precise}. In particular, procedure calls
and returns can be depicted as  ``$[_k$''- and ``$]_k$''-labeled edges in a graph.
Interprocedural static analyses need to ensure that the procedure calls
and returns are properly matched. Based on Lemma~\ref{lem:bmmstd1}, we have:

\begin{theorem}
Demand-driven interprocedural program analysis is BMM-hard.
\end{theorem}

\section{Related Work}\label{sec1:rw}
Despite extensive
work~\cite{FahndrichFSA98partial,HardekopfL07the,HeintzeT01aultra}, the
worst-case  complexity of inclusion-based pointer analysis remains
cubic~\cite{HardekopfL07the,SridharanF09the}. In the literature, many pointer
analyses have been formulated as a CFL-reachability
problem~\cite{Reps97program,SridharanGSB05demand,qirun13fast,ZhangXZYS14eff,ZhengR08demand}. Traditional
CFL-reachability algorithm also exhibits a cubic time
complexity~\cite{Reps97program}. The subcubic CFL-reachability algorithm was
proposed by~\citet{Chaudhuri08subcubic},  improving the cubic complexity by a factor of $\text{log
}n$. Asymptotically better algorithms exist for
special cases. For instance, Sridharan and Fink proposed a quadratic algorithm
when the input graph is restricted to be $k$-sparse~\cite{SridharanF09the}.  
\citet{Chaudhuri08subcubic} gave  
 $O(n^3 / \log^2 n)$-time and $O(n^\omega)$-time algorithms for  bounded-stack recursive state machines and hierarchical state
machines, respectively.
For
CFL-reachability-based approach, \citet{qirun13fast} proposed an $O(n+m\text{ log }m)$
algorithm for alias analysis if the underlying CFL is restricted to be a Dyck
language.  \citet{ZhangS17-context} also gave an $O(mn)$ time algorithm for
computing sound solutions for a class of interleaved Dyck-reachability.
The fastest algorithm for solving Dyck-Reachability is due to
\citet{ChatterjeeCP18opt} which runs in time $O(m + n\cdot \alpha(n))$ where
$\alpha(n)$ is the inverse Ackermann function.
When restricted to graphs with bounded treewidth, \citet{ChatterjeeGGIP19faster} gave faster
algorithms for solving demand-driven queries in the presence of graph
changes. When restricted to directed acyclic graphs, \citet{Yannakakis90graph} noted that
CFL-reachability could be solved in $O(n^\omega)$ time.
The work of \citet{McAllester02on} established a framework for determining the time
complexity of static analysis.

The work of \citet{ChatterjeeCP18opt} established a conditional cubic lower bound of
Dyck-Reachability. Their work gave a reduction from CFL parsing
~\cite{Lee02fast} which required a Dyck language of $k$ kinds of
parentheses.   
The class of \npda represents the languages (or problems) definable by a two way
nondeterministic pushdown automaton. \citet{AhoHU68time} showed that any problem
in the class \npda can be solved in cubic time. \citet{Rytter85fast} improved the cubic bound by a logarithmic factor leveraging
the well-known Four Russians' Trick~\cite{arlazarov1970on} to speed up set operations under the random
access machine (RAM) model.
In the work of \citet{HeintzeM97on}, it was shown that the
$D_2$-reachability problem  is \npda-complete --- it is both in \npda and \npda-hard.
Based on the Four Russians' Trick, many
subcubic algorithms have been proposed for solving program analysis problems
such as CFL-reachability~\cite{Chaudhuri08subcubic}, control flow analysis~\cite{Midtgaard-VanHorn:RR09} and pointer analysis~\cite{ZhangXZYS14eff}.
In the literature, the best known algorithm for solving the \npda-complete problems is  due to
\citet{Rytter85fast} which exhibits an  $O(n^3 /
\text{log }n)$ time complexity. 
Recently, \citet{2006.01491anders} gave an independent result on the
BMM-hardness of inclusion-based points-to analysis via different proof
techniques. \citet{2006.01491anders}'s reduction enables improved
algorithms on restricted cases. Our result based on $D_1$-reachability sheds
light on the hardness of analyzing unrestricted non-trivial C-style programs as
well as general
demand-driven interprocedural program-analysis problems. The two results offer complementary insights.



\section{Conclusion}\label{sec1:conc}
This paper has presented a formal proof to establish the hardness of
inclusion-based points-to analysis. Our result shows that it is unlikely to have
a truly subcubic time algorithm for inclusion-based points-to analysis in
practice. We have also discussed two interesting implications based on our reduction.

\bibliography{ref}

\end{document}